\PassOptionsToPackage{unicode}{hyperref}
\PassOptionsToPackage{hyphens}{url}
\PassOptionsToPackage{dvipsnames,svgnames,x11names}{xcolor}
\documentclass[
  article, nojss, nofooter]{jss}

\usepackage{amsmath,amssymb}
\usepackage{iftex}
\ifPDFTeX
  \usepackage[T1]{fontenc}
  \usepackage[utf8]{inputenc}
  \usepackage{textcomp} % provide euro and other symbols
\else % if luatex or xetex
  \usepackage{unicode-math}
  \defaultfontfeatures{Scale=MatchLowercase}
  \defaultfontfeatures[\rmfamily]{Ligatures=TeX,Scale=1}
\fi
\usepackage{lmodern}
\ifPDFTeX\else  
\fi
\IfFileExists{upquote.sty}{\usepackage{upquote}}{}
\IfFileExists{microtype.sty}{% use microtype if available
  \usepackage[]{microtype}
  \UseMicrotypeSet[protrusion]{basicmath} % disable protrusion for tt fonts
}{}
\makeatletter
\@ifundefined{KOMAClassName}{% if non-KOMA class
  \IfFileExists{parskip.sty}{%
    \usepackage{parskip}
  }{% else
    \setlength{\parindent}{0pt}
    \setlength{\parskip}{6pt plus 2pt minus 1pt}}
}{% if KOMA class
  \KOMAoptions{parskip=half}}
\makeatother
\usepackage{xcolor}
\ifx\paragraph\undefined\else
  \let\oldparagraph\paragraph
  \renewcommand{\paragraph}[1]{\oldparagraph{#1}\mbox{}}
\fi
\ifx\subparagraph\undefined\else
  \let\oldsubparagraph\subparagraph
  \renewcommand{\subparagraph}[1]{\oldsubparagraph{#1}\mbox{}}
\fi

\providecommand{\tightlist}{%
  \setlength{\itemsep}{0pt}\setlength{\parskip}{0pt}}\usepackage{longtable,booktabs,array}
\usepackage{calc} % for calculating minipage widths
\usepackage{etoolbox}
\makeatletter
\patchcmd\longtable{\par}{\if@noskipsec\mbox{}\fi\par}{}{}
\makeatother
\IfFileExists{footnotehyper.sty}{\usepackage{footnotehyper}}{\usepackage{footnote}}
\makesavenoteenv{longtable}
\usepackage{graphicx}
\makeatletter
\def\maxwidth{\ifdim\Gin@nat@width>\linewidth\linewidth\else\Gin@nat@width\fi}
\def\maxheight{\ifdim\Gin@nat@height>\textheight\textheight\else\Gin@nat@height\fi}
\makeatother
\setkeys{Gin}{width=\maxwidth,height=\maxheight,keepaspectratio}
\makeatletter
\def\fps@figure{htbp}
\makeatother

\usepackage{orcidlink}
\newtheorem{proposition}{Proposition}
\makeatletter
\@ifpackageloaded{caption}{}{\usepackage{caption}}
\AtBeginDocument{%
\ifdefined\contentsname
  \renewcommand*\contentsname{Table of contents}
\else
  \newcommand\contentsname{Table of contents}
\fi
\ifdefined\listfigurename
  \renewcommand*\listfigurename{List of Figures}
\else
  \newcommand\listfigurename{List of Figures}
\fi
\ifdefined\listtablename
  \renewcommand*\listtablename{List of Tables}
\else
  \newcommand\listtablename{List of Tables}
\fi
\ifdefined\figurename
  \renewcommand*\figurename{Figure}
\else
  \newcommand\figurename{Figure}
\fi
\ifdefined\tablename
  \renewcommand*\tablename{Table}
\else
  \newcommand\tablename{Table}
\fi
}
\@ifpackageloaded{float}{}{\usepackage{float}}
\floatstyle{ruled}
\@ifundefined{c@chapter}{\newfloat{codelisting}{h}{lop}}{\newfloat{codelisting}{h}{lop}[chapter]}
\floatname{codelisting}{Listing}

\makeatother
\makeatletter
\@ifpackageloaded{caption}{}{\usepackage{caption}}
\@ifpackageloaded{subcaption}{}{\usepackage{subcaption}}
\makeatother
\makeatletter
\@ifpackageloaded{tcolorbox}{}{\usepackage[skins,breakable]{tcolorbox}}
\makeatother
\makeatletter
\@ifundefined{shadecolor}{\definecolor{shadecolor}{rgb}{.97, .97, .97}}{}
\makeatother
\makeatletter
\ifdefined\Shaded\renewenvironment{Shaded}{\begin{tcolorbox}[frame hidden, breakable, enhanced, boxrule=0pt, borderline west={3pt}{0pt}{shadecolor}, interior hidden, sharp corners]}{\end{tcolorbox}}\fi
\makeatother
\ifLuaTeX
  \usepackage{selnolig}  % disable illegal ligatures
\fi
\usepackage{bookmark}

\IfFileExists{xurl.sty}{\usepackage{xurl}}{} % add URL line breaks if available
\hypersetup{
  pdftitle={ and : Fast Online Changepoint Detection in R and Python},
  pdfauthor={Gaetano Romano; Kes Ward; Yuntang Fan; Guillem Rigaill; Vincent Runge; Idris A. Eckley; Paul Fearnhead},
  pdfkeywords={changepoint detection, anomaly detection, online
algorithms, real-time analysis, functional pruning, R, Python},
  colorlinks=true,
  linkcolor={blue},
  filecolor={Maroon},
  citecolor={Blue},
  urlcolor={Blue},
  pdfcreator={LaTeX via pandoc}}

\author{Gaetano Romano~\orcidlink{0000-0002-7751-9017}\\Lancaster
University \And Kes Ward\\Lancaster University \AND Yuntang
Fan~\orcidlink{0009-0001-5958-5002}\\Lancaster University \And Guillem
Rigaill~\orcidlink{0000-0002-7176-7511}\\Université Paris-Saclay, CNRS,\\
INRAE, Univ Evry, \\  Institute of Plant Sciences Paris-Saclay
(IPS2) \&\\ Université Paris-Saclay, AgroParisTech,\\ INRAE, UMR MIA
Paris-Saclay \AND Vincent Runge~\orcidlink{0000-0002-4857-1799}\\LaMME,
Université Paris-Saclay,\\ CNRS, Univ Evry \And Idris A.
Eckley~\orcidlink{0000-0001-9957-2460}\\Lancaster University \AND Paul
Fearnhead~\orcidlink{0000-0002-9386-2341}\\Lancaster University}
\Plainauthor{Gaetano Romano, Kes Ward, Yuntang Fan, Guillem
Rigaill, Vincent Runge, Idris A. Eckley, Paul
Fearnhead} %% comma-separated

\Plainauthor{G. Romano, K. Ward, Y. Fan, G. Rigaill, V. Runge, I. A. Eckley, P. Fearnhead}

\title{\pkg{focus} and \pkg{focus-cpt}: Fast Online Changepoint
Detection in R and Python}
\Plaintitle{ and : Fast Online Changepoint Detection in R and
Python} %% without formatting

\Abstract{We present an \proglang{R} and \proglang{Python} package for
fast online changepoint detection in univariate and multivariate data
streams for a variety of models. The package implements the \pkg{focus}
family of algorithms, which compute the Generalised Likelihood Ratio
test for a single changepoint exactly and efficiently, with a
per-iteration cost of approximately \(\log(n)^d\) for a d-dimensional
sequence, without introducing approximations. This is achieved by
exploiting a connection between the location of the changepoint
candidates and the geometry of the data. The package supports a broad
range of models from the natural exponential family, including Gaussian,
Poisson, Binomial, Exponential and Gamma distributions, as well as a
non-parametric detector based on the empirical cumulative distribution
function and a detector for autoregressive data.}

\Keywords{changepoint detection, anomaly detection, online
algorithms, real-time analysis, functional
pruning, \proglang{R}, \proglang{Python}}
\Plainkeywords{changepoint detection, anomaly detection, online
algorithms, real-time analysis, functional pruning, R, Python}

\Address{
Gaetano Romano\\
School of Mathematical Sciences\\
Fylde College\\
Lancaster United Kingdom\\
E-mail: \email{g.romano@lancaster.ac.uk}\\
URL: \url{https://www.lancaster.ac.uk/maths/people/gaetano-romano}\\
\\~
Kes Ward\\
School of Mathematical Sciences\\
Fylde College\\
Lancaster United Kingdom\\
E-mail: \email{k.ward@lancaster.ac.uk}\\
\\~
Yuntang Fan\\
School of Mathematical Sciences\\
Fylde College\\
Lancaster United Kingdom\\
E-mail: \email{y.yuntang@lancaster.ac.uk}\\
\\~
Guillem Rigaill\\
23 Bd de France-Georges Pompidou, 91037 Évry-Courcouronnes\\
Orsay \& Palaiseau France\\
E-mail: \email{guillem.rigaill@inrae.fr}\\
\\~
Vincent Runge\\
23 Bd de France-Georges Pompidou, 91037 Évry-Courcouronnes\\
Gif-sur-Yvette France\\
E-mail: \email{vincent.runge@univ-evry.fr}\\
\\~
Idris A. Eckley\\
School of Mathematical Sciences\\
Fylde College\\
Lancaster United Kingdom\\
E-mail: \email{i.eckley@lancaster.ac.uk}\\
URL: \url{https://www.lancs.ac.uk/\textasciitilde eckley/}\\
\\~
Paul Fearnhead\\
School of Mathematical Sciences\\
Fylde College\\
Lancaster United Kingdom\\
E-mail: \email{p.fearnhead@lancs.ac.uk}\\
URL: \url{https://www.maths.lancs.ac.uk/\textasciitilde fearnhea/}\\
\\~

}

\begin{document}
\maketitle

\section{Introduction}\label{sec_introduction}

In the past decade, we have seen a resurgence of changepoint detection
approaches
\citep[in particular in the offline scenario, see][]{truong2020selective}.
The vast majority of real-world processes happen over time, and
methodologies that account for temporal shifts and anomalies in these
processes have been proven useful to many real world applications
\citep[e.g.,][]{shiryaev2002quickest, alvarez2020flight, xie2019asynchronous}.
Nowadays, with the advent of the Internet of Things, and the ubiquity of
large sensor networks that are able to monitor processes in real-time,
it is common to want to detect anomalies and changes online. For
example, real-time monitoring from engineering
\citep{alvarez2020flight, lobodzinski2025predictive}, cybersecurity
\citep{jeske2018statistical}, or self-driving cars
\citep{galceran2017multipolicy}, could require analysis of high
frequency data from potentially multiple streams to inform timely
decisions. Neuroscience \citep{jewell2020fast} and astrophysics
\citep{crupi2023searching} problems may involve processing hundreds of
observations per second to uncover novel biological and scientific
mechanisms.

Motivated by these problems, there has been a range of methodological
advances around fast and efficient statistical online changepoint
detection approaches, including, but not limited to
\cite{aue2024state, xie2023window, yu2023note, chen2022high}. Amongst
these, the \pkg{focus} algorithm, introduced in \cite{romano2023fast},
has found applications across a range of settings and has prompted
several methodological developments. These so far include an extension
to univariate exponential family models
\citep{ward2022poisson, ward2023constant}, an adaptation to
non-parametric settings \citep{romano2023log}, as well as a
generalisation to multivariate time series
\citep{pishchagina2025online}.

\pkg{focus} detects changes using the Generalised Likelihood Ratio (GLR)
test \citep{siegmund1995using}. The GLR evaluates whether the underlying
distribution generating the observations has changed at a given point in
time, by comparing the likelihood of the data under the null hypothesis
of no change against the alternative hypothesis of a change at some
unknown location, optimising over all possible changepoint locations.
This test has shown good statistical properties in terms of detection
power and false alarm control for several models
\citep[see][]{Lorden, basseville1993detection}. However, when analysing
\(n\) data points, a naive sequential implementation requires
\(\mathcal{O}(n)\) operations per time step, leading to an overall
\(\mathcal{O}(n^2)\) complexity that makes it unfeasible for true
real-time applications. The \pkg{focus} algorithm addresses this
limitation directly: it computes the GLR statistic exactly, without any
approximation, but does so with a per-iteration cost that grows only
logarithmically in the unvariate case with the length of the observed
sequence, and more generally at a \(\log(n)^d\) rate for a d-dimensional
sequence, as we describe in detail in Section \ref{sec_algorithm}.

Methodologies based on \pkg{focus} are already being trialled in diverse
areas, from detecting astrophysical events on satellites
\citep{dilillo2024gamma}, to monitoring IoT networks
\citep{yang2024communication} and digital infrastructure
\citep{romano2023log}. This paper and software package arise from the
need to provide a common interface to the \pkg{focus} family, that could
be of use to practitioners.

One particular aspect of this implementation is that the \proglang{R}
and \proglang{Python} interfaces share the same underlying
\proglang{C++} backend, implementing the algorithms, and therefore show
parity of functionalities and produce identical results. Across the two
interfaces, available on
\href{https://cran.r-project.org/web/packages/focus/index.html}{CRAN} as
\pkg{focus} and on
\href{https://pypi.org/project/focus-cpt/0.1.1/}{PyPI} as
\pkg{focus-cpt}, there are only minor differences in interface design
philosophy (e.g., the update mechanism is more functional in
\proglang{R}, while it is more object-oriented in \proglang{Python}).
For this reason, in the description of the interface in this article, we
will provide code through the \proglang{R} interface and highlight
differences in the \proglang{Python} interface where relevant in bold
(as most of the function calls are exactly the same). We will give a
detailed overview of what differs in Appendix \ref{sec_python_interface}
and have a couple of \proglang{Python} examples in Section
\ref{sec_real_time_applications}. \proglang{Python} users should
therefore not be discouraged from reading the main part of the article
as well. On the same note, a certain degree of repetition is somewhat
unavoidable in the article, as we want to provide examples of several
types of detectors across both interfaces. To save space, code dedicated
to plotting, printing, some tuning, and data preparation has been hidden
from the code chunks; however, it is fully available on GitHub at
\url{https://github.com/gtromano/unified_focus/tree/main/paper_code} as
a reproducible \href{https://quarto.org/}{Quarto} notebook.

The rest of the article is structured as follows. In Section
\ref{sec_algorithm}, we provide a high-level description of the
algorithms implemented in the package, and in Section
\ref{sec_interface}, we describe the interfaces in detail. In Section
\ref{sec_use_cases}, we provide simulated examples of common use cases
of the package. Finally, in Section \ref{sec_real_time_applications}, we
provide some examples of real-time applications of the package,
showcasing its potential in various domains.

\subsection{Related Software}\label{related-software}

A comprehensive overview of \proglang{R} and \proglang{Python} packages
for changepoint and anomaly detection is available through the CRAN Task
Views for Time Series Analysis \citep{ctv-timeseries} and Anomaly
Detection \citep{ctv-anomaly}; we refer the reader to these for a broad
survey. Here we restrict attention to statistical packages most closely
related to \pkg{focus}, focusing on likelihood-based methods and
omitting general-purpose anomaly detection toolkits.

For offline changepoint detection in \proglang{R}, we find
\pkg{changepoint} \citep{killick2014changepoint}, which implements the
PELT algorithm and binary segmentation for changes in mean, variance, or
both, in a univariate setting. For change-in-mean, implementing the fpop
and pDPA algorithms \citep{Maidstone, Rigaill}, we find the \pkg{fpopw}
package. A broad family of randomised and deterministic localised search
strategies is collected in the \pkg{breakfast} package
\citep{breakfast_pkg}, which includes wild binary segmentation, WBS2,
narrowest-over-threshold, and other methods. A further deterministic
alternative achieving near-linear computational cost while retaining
minimax-optimal estimation is seeded binary segmentation
\citep{kovacs2023seeded}, whose implementation is available on GitHub
only at \url{https://github.com/kovacssolt/SeedBinSeg}. The MOSUM-based
approach, which uses moving-sum statistics for multiple mean-change
estimation, is implemented in the \pkg{mosum} package \citep{Meier2021}.
For Bayesian offline analysis, \pkg{bcp} \citep{erdman2008bcp} provides
a model for changes in the mean of Gaussian data. In \proglang{Python},
the primary offline library is \pkg{ruptures}
\citep{truong2018ruptures}, which offers a modular and well-documented
interface covering PELT, binary segmentation, window-based methods, and
kernel cost functions across both parametric and non-parametric models.

The above packages are for detecting abrupt unstructured changes
(similar to \pkg{focus}). There are also packages that can detect
gradual or structured changes in offline data. These include the
\pkg{cpop} \citep{fearnhead2024cpop} package and the
\texttt{trandfilter} function \citep{kim2009ell_1} in the \pkg{genlasso}
\citep{arnold2016efficient} package for detecting changes-in-slope, and
the \pkg{gfpop} \citep{runge2020gfpop} package for detecting structured
changes.

For online (and sequential) changepoint detection in \proglang{R},
\pkg{ocd} \citep{chen2022ocd} implements a multiscale likelihood-ratio
procedure for high-dimensional Gaussian mean changes, with storage and
per-observation cost independent of the number of past observations, it
also provides access to several classical sequential procedures. A
Bayesian online changepoint detection algorithm is available through
\pkg{ocp} \citep{pagotto2019ocp}, which implements the algorithm of
\cite{adams2007bayesian}. In \proglang{Python}, the online changepoint
detection ecosystem is currently less consolidated than its offline
counterpart. An implementation of the same \cite{adams2007bayesian}
algorithm is available as a lightweight package on PyPI
(\url{https://pypi.org/project/bocd/}); \pkg{changepoint-online} is a
fully pythonic, early implementation of some algorithms included in
\pkg{focus-cpt} (available at
\url{https://pypi.org/project/changepoint-online/}). The proposed
implementation therefore occupies a distinct niche relative to this
landscape: unlike \pkg{ocd}, which is its closest competitor, it is
based on exact GLR statistics and naturally extends to exponential
family models and non-parametric settings
\cite[a comparison of the statistical performances of \pkg{focus} and \pkg{ocd} is available in][]{romano2023fast, pishchagina2025online};
and, crucially, compared to \pkg{changepoint-online} it provides a
unified \proglang{R} and \proglang{Python} interface backed by a shared
(and therefore easier to maintain) and faster \proglang{C++}
implementation.

\section{A Description of the Algorithm}\label{sec_algorithm}

\subsection{Model and the GLR
Statistics}\label{model-and-the-glr-statistics}

Consider a stream of observations \(\{ y_n\}_{n = 1, 2, \ldots}\), where
each \(y_n \in \mathbb{R}^d\), and define a subset of this time series
\(y_{a:b}\) as the data points \(\{ y_a, \ldots, y_b \}\) for integers
\(a < b\). We assume our observations follow a one-parameter exponential
family with density \[
\label{eq:exponential_family}
p(y_t \mid \theta) = \prod_{j = 1}^d \exp\!\left\{ -\frac{1}{2} \Big( A(\theta^{(j)}) - 2 \theta^{(j)} T(y_t^{(j)}) + B(y_t^{(j)}) \Big) \right\},
\] where \(T(y_t^{(j)}) \in \mathbb{R}\) is a sufficient statistic,
\(\theta^T = (\theta^{(1)}, \dots, \ \theta^{(d)}) \in \mathbb{R}^d\) is
the natural parameter. This formulation assumes independence across the
\(d\) univariate time series. Here
\(A(\theta)= (A(\theta^{(1)}),\ldots,A(\theta^{(d)}) ) ^ T \in \mathbb{R}^d\)
is a convex function of the natural parameter, and
\(B(y_t^{(j)})\in \mathbb{R}\) is a function of the data only. The
natural parameter, \(\theta\), is generally not the mean parameter
\(\mu\); for example, for Poisson data \(\theta = \log(\mu)\). The
second column of Table \ref{tab:expfam} shows the relationship
\(\theta = \theta(\mu)\) for each distribution supported by the package.
For the rest of this section,
\(\langle u, v \rangle = \sum_{k=1}^d u^{(k)} v^{(k)}\) will denote the
standard scalar product in \(\mathbb{R}^d\). Note that different models
can be employed for each dimension, allowing for more flexible
multivariate detection scenarios (this capability is illustrated in the
application of Section \ref{sec_meanvar}).

\begin{table}[t]
  \centering
  \resizebox{\textwidth}{!}{%
  \begin{tabular}{lllll}
    \hline
    Distribution              & $\theta(\mu)$           & $T(y)$ & $A(\theta)$          & $B(y)$ \\
    \hline
    Gaussian (change in mean)                  & $\mu$                   & $y$    & $\theta^2$           & $y^2 + \log(2\pi)$      \\
   Poisson                   & $\log\mu$               & $y$    & $2e^{\theta}$        & $2\log y!$ \\
    Binomial                  & $\log\frac{\mu}{1-\mu}$ & $y$    & $2m\log(1+e^\theta)$ & $-2\log \binom{m}{y}$   \\
    Gamma                     & $-k/\mu$                & $y$    & $-2k    \log(-\theta)$ & $-2(k-1) \log y + 2\log\Gamma(k)$ \\    
    Gaussian (change in var.) & $-1/(2\mu)$ & $y^2$ & $-\log(-\theta)$ & $\log(2\pi)$ \\                      
    \hline
  \end{tabular}
  }
  \caption{Examples of distributions from the natural exponential family supported by the package. The $\theta(\mu)$ column gives the natural parameter $\theta$ as a function of the mean parameter $\mu$; the remaining columns give the corresponding forms of the sufficient statistic $T(y)$ and the functions $A(\theta)$ and $B(y)$. Without loss of generality, the variance of the Gaussian change in mean model is assumed to be $1$, while the Gaussian change in variance model assumes that the mean is $0$. The number of trials is assumed to be $m$ for all dimensions in the Binomial model. For the Gamma model, the shape parameter $k$ is assumed to be known and equal across dimensions, and it is possible to obtain the Exponential by setting a Gamma with shape $k=1$.}
  \label{tab:expfam}
\end{table}

The objective of our algorithm is to, at time \(n\), decide whether we
have observed a change or not based on the current data \(y_{1:n}\)
obtained so far. To this end, our model assumes that each data point
\(y_t\) is drawn i.i.d. from a distribution with density
\(p(y_t \mid \theta_0)\) before the change and \(p(y_t \mid \theta_1)\)
after. Under this model, for any candidate changepoint
\(\tau \in \{1,\dots,n-1\}\), define the log-likelihood, assuming a
change at time \(\tau\), as \[
q_{\tau,n}(\theta_0,\theta_1)
= \sum_{t=1}^{\tau} \log p(y_t \mid \theta_0) + \sum_{t=\tau+1}^{n} \log p(y_t \mid \theta_1).
\] Then, at time \(n\), maximising over all possible changepoint
locations yields the envelope \[
\mathcal{Q}_n(\theta_0,\theta_1) = \max_{\tau\in\{1,\dots,n-1\}} q_{\tau,n}(\theta_0,\theta_1).
\] The generalised likelihood ratio (GLR) statistic for testing the
presence of a single changepoint in \(y_{1:n}\) is then \[
f_{\mathrm{GLR}}(y_{1:n}) = 2\left\{ \max_{\theta_0,\theta_1}\mathcal{Q}_n(\theta_0,\theta_1)
- \max_{\theta} \sum_{t=1}^{n} \log p(y_t \mid \theta)
\right\}.
\] The optimisation to obtain \(f_{\mathrm{GLR}}\) is over both
potential changepoint locations \(\tau\) and pre- and post-change
parameters \(\theta_0\) and \(\theta_1\). Notably, it is possible to
consider the case where the pre-change parameter \(\theta_0\) is known
by simply using a plug-in value for \(\theta_0\): this version of the
optimisation, in the Gaussian case, will lead to the well known
Page-CUSUM statistics \citep{Page}.

Evaluating the GLR statistic at every time step naively requires
\(\mathcal{O}(n)\) operations per step when processing the \(n\)th
observation, due to the maximum operator over \(n-1\) candidate change
locations, which is infeasible for real-time applications. A key insight
from \cite{romano2023fast, pishchagina2025online} is that it is possible
to reduce this set of locations without losing optimality, and this
operation is often referred to as \emph{pruning}. This is achieved by
decoupling the maximisation of the evaluation of the statistics
(e.g.~\(\max_{\theta_0,\theta_1}\)) from the search for candidate change
locations to consider. These two operations are reflected in the code
via the two functions:

\begin{enumerate}
\def\labelenumi{\arabic{enumi}.}
\tightlist
\item
  \texttt{detector\_update} which processes the new data point and then
  recovers a set of candidate changepoints
  \(D_n \subseteq \{1,\dots,n-1\}\) that are needed to maximise the GLR
  statistics.
\item
  \texttt{get\_statistics}, that evaluates the GLR statistic across the
  active candidates \(D_n\) for a given cost function from above.
\end{enumerate}

We will describe in the next section how we build and update the set of
candidates \(D_n\) sequentially, and how this is linked to the geometry
of the data.

\subsection{Underlying Geometry of the
Problem}\label{sec_update_pruning}

It is possible to relate the possible candidate change-point locations
to the location of points on the convex hull of the random walk of the
data \citep{romano2023fast,pishchagina2025online}. \pkg{focus} uses this
property to obtain \(D_n\). The vertices of this hull will be the only
points needed to consider for evaluation of
\(\mathcal{Q}_n(\theta_0,\theta_1)\). We now give a novel, generalised
illustration of this link between the geometry of data and the candidate
change locations.

We begin by defining

\[
\hat{\tau}_n = \operatorname*{arg\,max}_{\substack{\tau\in\{1,\dots,n-1\}}} \ \max_{\theta_0,\theta_1} \ q_{\tau,n}(\theta_0,\theta_1),
\] and, for a given value of \(\theta_0, \theta_1\), define the
maximiser of \(\mathcal{Q}_n(\theta_0,\theta_1)\) as:

\[
\hat{\tau}_n(\theta_0, \theta_1) =  \operatorname*{arg\,max}_{\substack{\tau\in\{1,\dots,n-1\}}} \Big\{ q_{\tau,n}(\theta_0,\theta_1) \Big\}.
\] Intuitively, \(\hat{\tau}_n(\theta_0, \theta_1)\) tells us which
candidate changepoint location is most supported by the data, given a
particular pair of pre- and post-change parameters. As we vary
\((\theta_0, \theta_1)\) over all possible values, different candidates
\(\tau\) will emerge as the maximiser, tracing out a partition of the
parameter space. Critically, many indices \(\tau \in \{1, \dots, n-1\}\)
will never be the maximiser for any value of \((\theta_0, \theta_1)\) as
their associated log-likelihood \(q_{\tau,n}(\theta_0,\theta_1)\) is
everywhere dominated by that of some other candidate. This allow to
precisely define the active candidates set \[
D_n = \{\tau \in \{1,\dots,n-1\} : \exists\, (\theta_0, \theta_1) \text{ such that }
\hat{\tau}_n(\theta_0,\theta_1) = \tau \}.
\] This set collects those candidates that contribute to
\(\mathcal{Q}_n(\theta_0,\theta_1)\).

At any future time \(n+1\),
\(\hat{\tau}_{n+1}(\theta_0, \theta_1) \in \{\hat{\tau}_{n}(\theta_0, \theta_1), n\}\),
e.g.~\(\hat{\tau}_{n+1}(\theta_0, \theta_1)\) can only be either the
current maximiser \(\hat{\tau}_n(\theta_0, \theta_1)\) or the new
candidate \(n\), so any \(\tau' \notin D_n\) cannot become a maximiser.
This means the active set can only grow by at most one element at each
step, \[
D_{n} \subseteq D_{n-1} \cup \{n\},
\] and any \(\tau' \in D_{n-1}\) that is not in \(D_{n}\) can be
permanently discarded from the optimisation.

To understand how we can build the set \(D_n\), we now look into the
geometry of our optimisation problem. Subtracting the constant
\(q_n(\theta_1) = \sum_{t=1}^{n} \log p(y_t \mid \theta_1)\) from
\(q_{\tau,n}(\theta_0,\theta_1)\) does not affect the argmax over
\(\tau\), and allows us to write \begin{align*}
\hat{\tau}_n(\theta_0, \theta_1) &= \operatorname*{arg\,max}_{\substack{\tau\in\{1,\dots,n-1\}}} \Big\{ q_{\tau,n}(\theta_0,\theta_1) - q_{n}(\theta_1)\Big\} \\
&= \operatorname*{arg\,max}_{\substack{\tau\in\{1,\dots,n-1\}}} \Big\{ \sum_{t=1}^{\tau} \Big[ \log p(y_t \mid \theta_0) - \log p(y_t \mid \theta_1)  \Big] \Big\} \\
&= \operatorname*{arg\,max}_{\substack{\tau\in\{1,\dots,n-1\}}} \Big\{ q_\tau(\theta_0, \theta_1) \Big\},
\end{align*}

Where
\(q_\tau(\theta_0, \theta_1) = \sum_{t=1}^{\tau} \Big[ \log p(y_t \mid \theta_0) - \log p(y_t \mid \theta_1)  \Big]\).
So \(q_\tau(\theta_0, \theta_1)\) and consequently
\(\hat{\tau}_n(\theta_0, \theta_1)\) do not depend on \(n\) anymore. By
defining: \[
P(\tau) = \left(\tau, \sum_{t=1}^\tau T(y_t)\right) \in \mathbb{R}^{d+1}
\] and \[
\psi(\theta_0, \theta_1) = \left(\frac{A(\theta_1) - A(\theta_0)}{2},\,
\frac{\theta_0 - \theta_1}{2}\right) \in \mathbb{R}^{d+1},
\] we can now write \(q_\tau\) as a product between a data-dependent
vector and a model dependent vector: \[
q_\tau(\theta_0, \theta_1) = \langle P(\tau),\, \psi(\theta_0, \theta_1) \rangle.
\] We note that maximising \(q_\tau\) over \(\tau\) is equivalent to
finding which point \(P(\tau)\) achieves the largest scalar product with
\(\psi(\theta_0, \theta_1)\), which brings us to the following
proposition.

\begin{proposition}
For any $\tau' \in \{1,\dots,n-1\}, \ n> d+2$, if $P(\tau')$ lies in the interior of the convex hull of $\{P(\tau)\}_{\tau \in \{1,\dots,n-1\}}$, then $\tau' \notin D_n$
and can be permanently pruned.
\end{proposition}

\emph{Proof} Since \(P(\tau')\) is in the interior of the convex hull,
there exist \(\tau_1,\dots,\tau_{d+2} \in \{1,\dots,n-1\}\) and weights
\(\lambda_1,\dots,
\lambda_{d+2} \geq 0\) with \(\sum_{i=1}^{d+2}\lambda_i = 1\) such that
\[
P(\tau') = \sum_{i=1}^{d+2} \lambda_i P(\tau_i).
\] For any \((\theta_0, \theta_1)\), taking the scalar product with
\(\psi(\theta_0,\theta_1)\) on both sides and using linearity gives
\begin{align*}
q_{\tau'}(\theta_0,\theta_1) &= \langle P(\tau'),\psi(\theta_0,\theta_1)\rangle = \sum_{i=1}^{d+2} \lambda_i \langle P(\tau_i), \psi(\theta_0,\theta_1)\rangle \\
&< \max_{i \in \{1,\ldots, d+2\}} \langle P(\tau_i), \psi(\theta_0,\theta_1)\rangle \leq \max_{\tau \in \{1,\ldots, n-1\}} q_\tau(\theta_0, \theta_1),
\end{align*}

where the strict inequality holds because \(P(\tau')\) is in the
interior of the hull. Therefore
\(q_{\tau'}(\theta_0,\theta_1) < \max_{\tau \in \{1,\dots,n-1\}} q_\tau(\theta_0,\theta_1)\)
for all \((\theta_0,\theta_1)\), meaning \(\tau'\) can never be the
maximiser \(\hat{\tau}_n(\theta_0,\theta_1)\) and so
\(\tau' \notin D_n\). \hfill \(\Box\)

This proposition shows that \(D_n\) is contained within the set of the
first component (the \(\tau\)s) of the vertices of the convex hull of
\(\{P(\tau)\}_{\tau \in \{1,\dots,n-1\}}\), and for some models, such as
the Gaussian change-in-mean, they are equivalent; see Theorem 2 of
\cite{pishchagina2025online}. It is this geometric structure that we
exploit to maintain \(D_n\) efficiently.

\subsection{Efficient Computations of Convex
Hulls}\label{efficient-computations-of-convex-hulls}

The way we build and update \(D_n\) in practice from \(D_{n-1}\) is by
iteratively integrating a new point
\(P(n) = \left(n, \sum_{t=1}^{n} T(y_t)\right)\) into the hull and
dropping points that lie on the interior of the hull, as they no longer
contribute to \(\hat{\tau}_n(\theta_0, \theta_1)\) (this is the pruning
step). How this is done differs between the univariate and multivariate
cases. In the univariate case, we employ the pruning algorithm described
in \cite{ward2023constant}, whose per iteration cost is amortised
\(\mathcal{O}(1)\) on average, and is closely related to the Melkman
algorithm \citep{melkman1987line}. This is based on comparisons of the
argmax of the component functions \(q_{i,n}\) and \(q_{i+1,n}\) for
pairs of contiguous elements of \(D_n\) and is equivalent to updating
sequentially a 2-dimensional hull by comparing the slope of contiguous
points. As updates of the majorants and minorants of the convex hull can
be kept separate, this also allows for constraining the detection to
only positive or negative changes; see Section \ref{sec_constrained}.

For multivariate data, standard convex hull algorithms can be applied to
prune candidates. As in \cite{pishchagina2025online}, we employ
QuickHull \citep{barber1996quickhull} for reconstructing the set of
optimal points. Because running a full convex hull algorithm such as
QuickHull at every iteration would itself be costly, the implementation
runs it only intermittently. Specifically, the algorithm maintains a
working candidate set \(T'_n\) that is allowed to grow beyond the true
hull vertex set \(D_n\); QuickHull is triggered only once \(|T'_n|\)
exceeds a size threshold, which is itself updated after each pruning
step as \(\lfloor \alpha |T'_{n}| + \beta \rfloor\) for user-controlled
parameters \(\alpha \geq 1\), the pruning multiplier, and
\(\beta \geq 0\), the pruning offset. This means that the algorithm
always retains a small number of extra candidates beyond those strictly
necessary, trading a small increase in the candidate set size for a
significant reduction in the number of times the hull algorithm is
called. In practice, setting \(\alpha = 2\) and \(\beta = 1\) is found
to yield good performance across a range of dimensions and sequence
lengths, but these values can be overridden as described in Section
\ref{sec_detectors}.

One feature of multidimensional sequences, as illustrated in Corollary 1
of \cite{pishchagina2025online}, is that the expected size of \(D_n\)
grows only as \(\mathcal{O}(\log (n)^d )\) in the \(d\)-dimensional
case, rather than linearly with \(n\). So for low-dimensional
multivariate data, such as \(d \leq 5\), one can use the detector
corresponding to the full \((d+1)\)-dimensional convex hull. For
higher-dimensional problems, this becomes computationally impractical
due to the \(\log(n)^d\) growth in the number of hull vertices, and an
approximate algorithm is available: the data are projected onto
overlapping subsets of \(\tilde{d} \leq 5\) coordinates, a convex hull
is computed on each lower-dimensional projection, and the union of the
resulting vertex sets is taken as the set of candidates. This
approximation stores \(\mathcal{O}(d \log(n)^{\tilde{d}})\) candidates
rather than \(\mathcal{O}(\log(n)^d)\), substantially reducing the
computational burden for large \(d\) while retaining good statistical
performance \citep{pishchagina2025online}. How to use this, and the
computational gains obtained from the approximations are briefly
exemplified in Section \ref{sec_multi_detectors}.

\subsection{Non-Parametric and Autoregressive
Detectors}\label{non-parametric-and-autoregressive-detectors}

The package also implements two extensions beyond the i.i.d. exponential
family framework described above. We describe each in turn.

\subsubsection{Non-Parametric Detector}\label{sec_nonparametric_method}

In practice, the underlying distribution of the data or the pre- and
post-change distributions are often unknown, making the parametric
models of Table \ref{tab:expfam} unsuitable for certain scenarios. For
this reason, the package implements the \pkg{npfocus} detector of
\cite{romano2023log}, which detects changes in the empirical cumulative
distribution function (eCDF) of a univariate data stream without any
distributional assumptions. The key idea is to approximate the eCDF on a
fixed grid of \(M\) quantile values \(m_1, \dots, m_M\), computed either
from historical knowledge or over a training period, and to reduce the
problem to \(M\) independent Binomial changepoint detection problems:
for each quantile \(m_j\), the original observations \(y_t\) are
replaced by the binary sequence
\(z^{(j)}_t = \mathbb{I}(y_t \leq m_j)\), and a Binomial \pkg{focus}
detector is run on each such sequence.

Consistent with the general \pkg{focus} architecture, we find a
separation between pruning the changepoint set and the model evaluation:
for each quantile \(m_j\), the \texttt{detector\_update} function
independently maintains the set of candidate changepoints \(D_n^{(j)}\)
based on the binary sequence \(z^{(j)}_{1:n}\), using the same convex
hull geometry principle described in Section \ref{sec_update_pruning};
the \texttt{get\_statistics} function then evaluates the Binomial GLR
statistic at each of these candidates. A global test statistic is then
constructed by aggregating the \(M\) individual GLR statistics either
through their sum or their maximum. A changepoint is declared at time
\(n\) if \[
\sum_{j=1}^M f^{(j)}_{\mathrm{GLR}}(z^{(j)}_{1:n}) \geq \xi^{\mathrm{sum}}
\quad \text{or} \quad
\max_{j \in \{1,\dots,M\}} f^{(j)}_{\mathrm{GLR}}(z^{(j)}_{1:n}) \geq \xi^{\mathrm{max}},
\] for user-specified thresholds
\(\xi^{\mathrm{sum}}, \xi^{\mathrm{max}} \in \mathbb{R}\), where
\(f^{(j)}_{\mathrm{GLR}}(z^{(j)}_{1:n})\) denotes the GLR statistic
applied to the \(j\)-th quantile transformation. The user can choose to
perform either both or one of the two tests, depending on what type of
changes they wish to detect in the distribution: the sum statistic is
more powerful for detecting changes that affect multiple quantiles
(e.g., a shift in the scale or location), while the max statistic is
more powerful for detecting changes that affect only a few quantiles
(such as a change in the tail of the distribution). Section
\ref{sec_nonparametric} provides an example of this.

\subsubsection{Autoregressive Detector}\label{sec_ar_method}

Many real-world time series exhibit temporal dependencies, where
observations at time \(t\) depend on observations at previous times.
Unlike the i.i.d. detectors described above, to deal with temporal
dependence in univariate sequences, the package also implements a
detector for autoregressive data, based on the work of
\cite{fan2026efficientlikelihoodratiotest}.

In this case, we assume the following AR(\(p\)) model: \[
y_t = \mu^{(t)} + \epsilon_t, \quad \text{where} \quad \epsilon_t = \sum_{i=1}^p \rho_i \epsilon_{t-i} + \xi_t,
\] where \(\mu^{(t)} = \mu_0\) for \(t \leq \tau\) and
\(\mu^{(t)} = \mu_1\) for \(t > \tau\), and
\(\xi_t \sim \mathcal{N}(0, \sigma_\xi^2)\). The AR coefficients
\(\rho_1, \ldots, \rho_p\) and the variance \(\sigma_\xi^2\) are assumed
to be known or estimated over a training period; see Section
\ref{sec_ar} for an example of how this can be achieved.

A crude approximation would be to pre-whiten the data using the AR
coefficients and then apply the i.i.d. \pkg{focus} algorithm to the
residuals \(\xi_{1:n}\). However, this approach has a fundamental
limitation: it ignores that immediately after a changepoint at time
\(\tau\), the mean of the pre-whitened data is non-constant over the
first \(p\) observations. Specifically, the first \(p\) observations
after a changepoint exhibit transient behaviour because the AR process's
lagged dependencies still include pre-change values, and the fitted
residuals only become constant once these lags have fully transitioned
to post-change values. Ignoring this transient can lead to a loss of
power for detecting changes that occur near the end of the stream.

The AR(\(p\))-\pkg{focus} algorithm instead implements the GLR for the
AR(\(p\)) model exactly. This is achieved by separately computing the
GLR for changepoint locations \(\tau\) within the last \(p\)
observations and for locations beyond \(p\) observations before \(n\).
The algorithm then applies a slightly modified \pkg{focus} pruning. For
brevity, we do not describe the resulting \(\mathcal{Q}_n(\mu_0,\mu_1)\)
and pruning in detail, but these come from a direct extension of the
original Gaussian algorithm, as described in Section 3 of
\cite{fan2026efficientlikelihoodratiotest}. The
\texttt{detector\_update} and \texttt{get\_statistics} functions
maintain the same separation by managing the candidate set and
evaluating the AR(\(p\)) GLR statistic, respectively. In the case of
known pre-change parameters, notably, unlike the i.i.d. case, the
decoupling of the cost and pruning is not directly possible, and the
pruning step requires knowledge of the pre-change mean \(\mu_0\) in
addition to the post-change \(\mu_1\). For practical details on how this
is handled in the interface and for usage examples, see Section
\ref{sec_ar}.

\section{The Interface}\label{sec_interface}

We now provide a description of the interface of the \pkg{focus} package
in \proglang{R} (and highlight differences for the \proglang{Python} one
where relevant, and additionally in Appendix
\ref{sec_python_interface}). The package presents two different modes of
operation: an online (sequential) mode and an offline (batch) mode. The
online mode is designed for real-time applications where data arrives
sequentially and offers the possibility to update the changepoint
detector with new data points as they arrive, as well as to retrieve the
statistics and the current changepoint estimates at every iteration.
Where the entire data set is available at once, a special function
called \texttt{focus\_offline} allows to processing a time-series as a
whole, and it is intended for scenarios such as training and testing,
benchmarking, or retrospective analysis (see Section
\ref{sec_offline_interface} for more details). As
\texttt{focus\_offline} is performing the same operations and producing
the same results as the online interface, in the following code
snippets, and for most of this section, we focus on the online mode, as
it is more flexible and more commonly used in real-time applications.

\textbf{A Quick Example.} Let us start by generating some Gaussian data
with a change in the mean:

\begin{verbatim}
R> set.seed(42)
R> Y <- c(rnorm(100, mean = 0), rnorm(50, mean = 1))
\end{verbatim}

For running the algorithm online, the first step is to create a detector
object using the \texttt{detector\_create} function. This object will
hold the state of the detector and will be updated with new data points
as they arrive. In this example, we create a univariate detector, but
will provide examples for multivariate detectors in Section
\ref{sec_multi_detectors}. For comparison purposes, in the snippet
below, we store the trace of the statistics in a vector called
\texttt{stat\_trace}. Finally, we run the main loop of the online
interface. In a real-time application, this loop would be running
indefinitely, and new data points would be processed as they arrive.

\begin{verbatim}
R> library(focus)
R> det <- detector_create(type = "univariate")
R> for (i in seq_along(Y)) {
+    detector_update(det, Y[i])
+    result <- get_statistics(det, family = "gaussian")
+    if (result$stat > 20)
+      break
+  }
R> sprintf("Changepoint detected at time %d", i)
\end{verbatim}

\begin{verbatim}
[1] "Changepoint detected at time 141"
\end{verbatim}

Within the \texttt{for} loop, we update the detector with a new data
point with \texttt{detector\_update} and retrieve the current statistics
\texttt{get\_statistics}; both functions will be explained in detail
below. \texttt{get\_statistics} returns a list, where \texttt{\$stat} is
our GLR test statistic: we break the cycle as soon as this is over a
fixed threshold of 20, and print the stopping time. The online interface
supports modern \proglang{R} piping, so the update and retrieval of
statistics can be chained together in a more concise way. For example,
the following snippet is equivalent to the above calls to
\texttt{detector\_update} and \texttt{get\_statistics}:

\begin{verbatim}
R> result <- det |>
+    detector_update(Y[i]) |>
+    get_statistics(family = "gaussian")
\end{verbatim}

This is possible thanks to the fact that both functions return the
detector object itself, allowing for the chaining of operations. This is
a design choice that we made to enhance the usability of the online
interface to align with modern \proglang{R} programming practices. In
practice, \texttt{detector\_update} updates a smart pointer to the
internal state of the detector (see full description below for details),
so the overhead of returning the detector object is negligible.

\textbf{For \proglang{Python} readers.} The pipe operator
\texttt{\textbar{}\textgreater{}} syntax in \proglang{R} will be useful
as well to aid reading for \proglang{Python} users. This is because it
is more syntactically similar to the method chaining in
\proglang{Python}, so most of the code can be directly adapted to
\proglang{Python} by simply replacing the pipe operator with method
calls on the detector object, e.g.:

\begin{verbatim}
>>> result = (det
>>>           .detector_update(Y[i])
>>>           .get_statistics(family="gaussian"))
\end{verbatim}

For the rest of the section, we will begin by introducing the different
types of detectors and cost functions that are implemented in the
package, and then we will describe the offline interface in detail.

\subsection{Changepoint Detectors}\label{sec_detectors}

A detector is the core object of the package, and it holds the current
state of the changepoint detection algorithm. The detector stores and
handles the update of the optimal changepoint locations, while the cost
computations are handled by the cost functions. The cost functions are
designed to be modular and flexible, allowing for a wide range of models
(even for the same detector), as well as user-defined cost functions
(see Section \ref{sec_extracting_info} and \ref{sec_meanvar} for
guidance and an example respectively). Note that not all detectors are
however compatible with all cost functions, but roughly the following
combinations are implemented:

\begin{itemize}
\tightlist
\item
  Detectors for iid data: these work for univariate and multivariate
  data, and they are compatible with many cost functions for exponential
  family distribution models.
\item
  Non-parametric detector: this is designed for univariate data, and
  implements the non-parametric cost function described in
  \citet{romano2023log}. An example of this will be specifically
  described in Section \ref{sec_nonparametric}.
\item
  AR(p) detector: this is designed for univariate data, and it is
  compatible with cost functions for autoregressive models. Again, as
  this detector is tied to a specific cost function, an example of this
  will be described in Section \ref{sec_ar}.
\end{itemize}

\subsubsection{Initialise a New
Detector.}\label{initialise-a-new-detector.}

To initialise a detector, one can simply call \texttt{detector\_create}:

\begin{verbatim}
detector_create(type, dim_indexes = NULL, quantiles = NULL,
  pruning_mult = 2L, pruning_offset = 1L, side = "right",
  anomaly_intensity = NULL, rho = NULL, mu0_arp = NULL
)
\end{verbatim}

\textbf{For \proglang{Python} readers.}, the \proglang{Python} class
leverages a constructor to initialize a detector object. Instead of
calling \texttt{detector\_create(type\ =\ "univariate",\ ...)}, one
creates a \texttt{Detector} instance with
\texttt{Detector(type="univariate",\ ...)}. All arguments are equivalent
to the \proglang{R} function, but the returned object is a
\proglang{Python} class instance rather than an external pointer. The
other major difference is that the default \texttt{NULL} values in
\proglang{R} are replaced by \texttt{None} in \proglang{Python}.

The main arguments to \texttt{detector\_create} are as follows:

\begin{itemize}
\tightlist
\item
  \texttt{type}: Specifies the type of detector to create. Options
  include:

  \begin{itemize}
  \tightlist
  \item
    \texttt{\textquotesingle{}univariate\textquotesingle{}}: Standard
    two-sided (e.g.~for both increases and decreases in the parameter
    signal) univariate detection.
  \item
    \texttt{\textquotesingle{}univariate\_one\_sided\textquotesingle{}}:
    One-sided univariate detection (detects only increases or
    decreases). Which change to detect is determined by the
    \texttt{side} argument (see below).
  \item
    \texttt{\textquotesingle{}multivariate\textquotesingle{}}:
    Multivariate detection using projections (for higher-dimensional
    data).
  \item
    \texttt{\textquotesingle{}npfocus\textquotesingle{}}: Nonparametric
    detection (see Section \ref{sec_nonparametric}).
  \item
    \texttt{\textquotesingle{}arp\textquotesingle{}}: Detection for
    autoregressive processes (see Section \ref{sec_ar}).
  \end{itemize}
\item
  \texttt{dim\_indexes}: For high-dimensional multivariate detection, a
  list of integer vectors specifying which dimensions to use for each
  projection (see example in Section \ref{sec_multi_detectors}). For
  low-dimensional data (5 dimensions or fewer), this is generally not
  required.
\item
  \texttt{quantiles}: Numeric vector of quantiles to track the empirical
  distribution over, required for nonparametric detectors
  (\texttt{type\ =\ \textquotesingle{}npfocus\textquotesingle{}}).
\item
  \texttt{pruning\_mult}, \texttt{pruning\_offset}: Parameters
  controlling the number of iterations before multivariate pruning is
  performed, see Section \ref{sec_update_pruning}
\item
  \texttt{side}: For
  \texttt{\textquotesingle{}univariate\_one\_sided\textquotesingle{}}
  detectors only, determines the direction of detection. Use
  \texttt{\textquotesingle{}right\textquotesingle{}} to detect
  increases, or \texttt{\textquotesingle{}left\textquotesingle{}} to
  detect decreases. Default is
  \texttt{\textquotesingle{}right\textquotesingle{}} (ignored in case of
  any other detector outside
  \texttt{\textquotesingle{}univariate\_one\_sided\textquotesingle{}}).
\item
  \texttt{anomaly\_intensity}: Numeric threshold for pruning candidates
  based on signal strength. If set, only candidates with sufficient
  magnitude are retained. See Section \ref{sec_anomaly_detection} for an
  example. By default, this is disabled.
\item
  \texttt{rho}: Numeric vector of AR coefficients, required for
  autoregressive detectors
  (\texttt{type\ =\ \textquotesingle{}arp\textquotesingle{}}).
\item
  \texttt{mu0\_arp}: Optional pre-change mean for AR(\(p\)) detectors.
  Whilst normally the pruning should be independent of the pre-change
  parameter (and for all models is passed with \texttt{get\_statistics},
  see Section \ref{sec_statistics}), this is not the case for the
  AR(\(p\)) model. Therefore, this to be specified here if known. An
  example is in Section \ref{sec_ar}.
\end{itemize}

The detector object returned by \texttt{detector\_create} is passed to
subsequent functions to update its state and retrieve statistics.

\subsubsection{Update the Detector}\label{update-the-detector}

To update a detector at every iteration, one should call
\texttt{detector\_update} with the new data point, and the detector
object from the previous iteration.

\begin{verbatim}
detector_update(det_ptr, y, lambda=1)
\end{verbatim}

\textbf{For \proglang{Python} readers.}, the update is performed by
calling the \texttt{update} method of the detector object, so no
detector argument is needed (e.g.~\texttt{detector.update(y)}). Note
that the \texttt{lambda} argument is called \texttt{lambda\_}.

Its arguments are:

\begin{itemize}
\tightlist
\item
  \texttt{det\_ptr}: The detector object (an external smart pointer)
  created by \texttt{detector\_create}.
\item
  \texttt{y}: The new observation(s) to add. For univariate detectors,
  this is a single numeric value. For multivariate detectors, this
  should be a numeric vector matching the number of dimensions.
\item
  \texttt{lambda}: A weight for the current observation, defaulting to
  \texttt{1}. Rather than incrementing the internal time counter by 1 at
  each step, the counter is incremented by \texttt{lambda}, so that
  \(P(\tau)\) for all \(\tau \in D_n\) (Section
  \ref{sec_update_pruning}) becomes \(\sum_{i=1}^{\tau} \lambda_i\).
  This enables non-stationary models where the parameter of interest
  scales with a known covariate: for instance, a Gaussian observation
  \(y_i \sim \mathcal{N}(\lambda_i \theta, \nu_i)\) or a Poisson
  observation \(y_i \sim \text{Poisson}(\lambda_i \theta)\) with a
  time-varying null (background) rate \(\lambda_i\). This may be useful
  when implementing custom cost functions, see Section
  \ref{sec_meanvar}.
\end{itemize}

The detector is updated in place, so the same object is returned
(allowing for chaining with the pipe operator if desired).

For example, here a detector is initialised and updated with a couple of
observations:

\begin{verbatim}
R> det_uni <- detector_create(type = "univariate")
R> det_uni <- det_uni |> detector_update(0.5) 
R> det_uni <- det_uni |> detector_update(0.1)
R> print(det_uni)
\end{verbatim}

\begin{verbatim}
<pointer: 0x5667f50b71b0>
\end{verbatim}

The output \texttt{\textless{}pointer:\ 0x...\textgreater{}} indicates
that the detector object is an external pointer to the underlying C++
implementation (printing the hexadecimal address uniquely identifies the
detector instance being referenced). It is important to note that
detector objects maintain mutable state, meaning that all variables
pointing to the same detector instance share that state. For example, if
you assign the result of \texttt{detector\_update} to different
variables (e.g.,
\texttt{det\_uni1\ \textless{}-\ detector\_update(det\_uni,\ 0.5)} and
\texttt{det\_uni2\ \textless{}-\ detector\_update(det\_uni1,\ 0.1)}),
both \texttt{det\_uni1} and \texttt{det\_uni2} reference the same
underlying C++ object. In \proglang{R}, the return is not strictly
necessary to update the detector, as the update is performed in place,
so one could simply call \texttt{detector\_update(det\_uni,\ 0.5)}
without assignment, and the state of \texttt{det\_uni} would still be
updated. However, we return the detector object for consistency with the
pipe operator and to allow for chaining of operations.

\subsection{Compute Statistics}\label{sec_statistics}

As explained in Section \ref{sec_algorithm}, the
\texttt{get\_statistics} function computes the log-likelihood ratio test
statistic from the set of optimal changepoints at the current time.
Therefore, in a normal workflow, at every iteration, one should call
\texttt{get\_statistics} with the updated detector object and the cost
function parameters, and compare the resulting statistic against a
threshold to decide whether to declare a changepoint.

\begin{verbatim}
get_statistics(det_ptr, family, theta0 = NULL, shape = NULL)
\end{verbatim}

\textbf{For \proglang{Python} readers.}, the statistics are retrieved by
calling the \texttt{get\_statistics} method of the detector object, so
no detector argument is needed
(e.g.~\texttt{detector.get\_statistics(family,\ theta0=None,\ shape=None)}),
with the arguments described below.

The main arguments are:

\begin{itemize}
\tightlist
\item
  \texttt{det\_ptr}: The detector object created by
  \texttt{detector\_create}.
\item
  \texttt{family}: Character string specifying the distribution family.
  Options include:

  \begin{itemize}
  \tightlist
  \item
    \texttt{\textquotesingle{}gaussian\textquotesingle{}}: Gaussian
    (normal) distribution for continuous data.
  \item
    \texttt{\textquotesingle{}poisson\textquotesingle{}}: Poisson
    distribution for count data.
  \item
    \texttt{\textquotesingle{}bernoulli\textquotesingle{}}: Bernoulli
    distribution for binary data.
  \item
    \texttt{\textquotesingle{}gamma\textquotesingle{}}: Gamma
    distribution (requires \texttt{shape} parameter).
  \item
    \texttt{\textquotesingle{}npfocus\textquotesingle{}}: Nonparametric
    detection (see Section \ref{sec_nonparametric}).
  \item
    \texttt{\textquotesingle{}arp\textquotesingle{}}: AutoRegressive
    Process detection (see Section \ref{sec_ar}).
  \end{itemize}
\item
  \texttt{theta0}: Numeric vector specifying the null hypothesis
  parameter if known (the \(\theta_0\) pre-change parameter). For
  univariate detectors, this is a scalar; for multivariate detectors, a
  vector matching the number of dimensions. Default is \texttt{NULL}, in
  which case the null parameter is estimated.
\item
  \texttt{shape}: Numeric scalar for the shape parameter required by the
  gamma distribution. This is only used when
  \texttt{family\ =\ \textquotesingle{}gamma\textquotesingle{}}.
\end{itemize}

The function returns a list (or a dictionary, in \proglang{Python}) with
the following components:

\begin{itemize}
\tightlist
\item
  \texttt{stopping\_time}: Integer indicating the current time index
  (number of observations processed).
\item
  \texttt{changepoint}: The current best estimate of the changepoint
  location, if there is a change. This is the \(\hat\tau_n\) defined in
  Section \ref{sec_update_pruning}.
\item
  \texttt{stat}: Numeric scalar or vector. The test statistic(s). This
  is either a scalar value (e.g.~in univariate or multivariate), or a
  vector (for instance, in case of an \texttt{npfocus} detector).
\end{itemize}

An important aspect of the implementation is that one can test multiple
cost functions on the same detector object, as the statistics
computation is independent of the detector state. This allows for model
comparison without needing to rerun the detector.

For example, on the detector initialised and updated in the previous
section, below we compute the Gaussian and Gamma statistics:

\begin{verbatim}
R> get_statistics(det_uni, family = "gaussian")$stat
\end{verbatim}

\begin{verbatim}
[1] 0.08
\end{verbatim}

\begin{verbatim}
R> get_statistics(det_uni, family = "gamma", shape = 2)$stat
\end{verbatim}

\begin{verbatim}
[1] 1.175573
\end{verbatim}

\subsection{Extracting Informations from the Detector
Object}\label{sec_extracting_info}

Beyond computing the test statistics, the package provides several
functions to inspect and extract information from the detector object.
These functions allow users to access the internal state of the
detector, and it can be particularly useful to implement new cost
functions without reimplementing the core pruning algorithm. For the
rest of this subsection, we will describe the main functions for
extracting information from the detector object, but for a practical
implementation of a new cost, one should refer to Section
\ref{sec_meanvar}, where we provide an example of implementing a custom
cost function for Gaussian change in mean and variance detection. To
understand the returned values of these functions, it is useful to
recall the notation of Section \ref{sec_update_pruning}.

\subsubsection{Accessing Changepoint
Candidates}\label{sec_accessing_candidates}

The detector maintains a set of candidate changepoint and summary
statistics. Two functions are provided to access this information:

\begin{verbatim}
detector_candidates(det_ptr)
detector_cands_len(det_ptr)
\end{verbatim}

\texttt{detector\_candidates} returns a detailed data frame with
information about all candidate segments. To draw a parallel with
Section \ref{sec_update_pruning}, at time \(n\) this function will
expose the internal \(P(\tau)\) for all \(\tau \in D_n\). It will return
a \texttt{DataFrame} in \proglang{R} (or a list of lists in
\proglang{Python}). Each row represents a candidate segment from a
changepoint location to the current time. The data frame contains the
following columns:

\begin{itemize}
\tightlist
\item
  \texttt{tau}: Integer vector of candidate changepoint locations,
  e.g.~the first component of \(P(\tau)\).
\item
  \texttt{st}: List of numeric vectors containing the sufficient
  statistics for each candidate segment. These statistics (e.g.,
  cumulative sums, cumulative sum of squares), in conjunction with the
  functions below, are used to efficiently compute test statistics
  without reprocessing the raw data, corresponding to the remaining
  \(d+1\) components of \(P(\tau)\).
\item
  \texttt{side}: Character vector indicating the side for each
  candidate. This is which side of the univariate hull are the
  candidates related to, either the convex majorant or minorant. It
  could be relevant for implementation of one-sided cost functions, but
  it can be ignored for multivatiate detectors.
\end{itemize}

\textbf{For \proglang{Python} readers.}, the candidate information
functions are available as methods: \texttt{detector.get\_candidates()}
returns the candidate segments (as a dictionary) and
\texttt{detector.get\_n\_candidates()} returns the number of candidate
segments.

Here is an example:

\begin{verbatim}
R> candidates <- detector_candidates(det)
R> n_cand <- detector_cands_len(det)
R> print(n_cand)
\end{verbatim}

\begin{verbatim}
[1] 12
\end{verbatim}

\begin{verbatim}
R> print(head(candidates))
\end{verbatim}

\begin{verbatim}
  tau        st  side
1   0         0 right
2  47 -3.451892 right
3  61 -1.966885 right
4  99  2.598277 right
5 100  3.251482 right
6 129  24.42906 right
\end{verbatim}

\texttt{detector\_cands\_len} returns an integer indicating the number
of candidate segments currently tracked by the detector. As the detector
processes new observations, this number may grow or shrink depending on
the pruning strategy (e.g.~in the multivariate case, the number of
candidates at a given iteration is controlled by the
\texttt{pruning\_mult} and \texttt{pruning\_offset} parameters set at
detector creation, which could be also used to compute worst-case run
time scenarios in case of time-sensitive applications).

\subsubsection{Number of Observations and Cumulative
Statistics}\label{number-of-observations-and-cumulative-statistics}

The following two functions provide access to the sufficient statistics
for the full data sequence. These are needed to efficiently compute
(\(\mathcal{O}(1)\) complexity) the log-likelihood values
\(q_{\tau,n}(\theta_0,\theta_1)\) of user defined functions given
candidate changepoints \(\tau \in D_n\) without reprocessing the raw
data via telescopic sum cancellations
(i.e.~\(\sum_{t=1}^n T(y_t) - \sum_{t=1}^\tau T(y_t) = \sum_{t=\tau+1}^n T(y_t)\)).

\begin{verbatim}
detector_info_n(det_ptr)
detector_info_sn(det_ptr)
\end{verbatim}

Specifically, \texttt{detector\_info\_n} returns the current number of
observations processed, which corresponds to the time index \(n\). This
indicates how many data points have been processed since the detector
was created. \texttt{detector\_info\_sn} returns the full cumulative sum
statistics maintained internally by the detector,
e.g.~\(\sum_{t=1}^n T(y_t)\). For univariate detectors, this returns a
scalar; for multivariate detectors, this returns a vector of length
equal to the number of dimensions.

\textbf{For \proglang{Python} readers.}, these functions are available
as methods of the detector object, so they can be called as
\texttt{detector.get\_n()} and \texttt{detector.get\_sn()},
respectively.

For example:

\begin{verbatim}
R> n <- detector_info_n(det)
R> cumsum_stat <- detector_info_sn(det)
R> sprintf("n")
\end{verbatim}

\begin{verbatim}
[1] "n"
\end{verbatim}

\begin{verbatim}
R> sprintf("%.3f", cumsum_stat)
\end{verbatim}

\begin{verbatim}
[1] "38.625"
\end{verbatim}

\subsection{The Offline Interface}\label{sec_offline_interface}

\texttt{focus\_offline} runs the complete changepoint detection
algorithm over a whole data set in a single call. The primary use of
\texttt{focus\_offline} is for computation of the threshold, as seen
later in the examples of Sections \ref{sec_use_cases}, and for
retrospective analysis. Compared to the online interface, the offline
call avoids repeated \proglang{R}-to-C++ crossings and is therefore
(slightly) more efficient.

\begin{verbatim}
focus_offline(Y, threshold, type = "univariate", family = "gaussian",
  theta0 = NULL, dim_indexes = NULL, quantiles = NULL,
  pruning_mult = 2L, pruning_offset = 1L, side = "right",
  shape = NULL, anomaly_intensity = NULL, rho = NULL, mu0_arp = NULL
)
\end{verbatim}

\textbf{For \proglang{Python} readers.}, \texttt{focus\_offline} is
available as a stand-alone function with the same signature and
behaviour as in \proglang{R}.

The \texttt{focus\_offline} function takes as input the observed data
\texttt{Y} (a numeric vector for univariate problems or an
\(\mathbb{R}^{n \times d}\) matrix for multivariate problems), a
detection threshold, and a set of detector configuration and cost
function parameters. Since most arguments of \texttt{focus\_offline}
mirror those of the online interface (combining the arguments of
\texttt{detector\_create} and \texttt{get\_statistics}), we only
describe the \texttt{threshold} argument. For other detector type and
cost function parameters, refer to Sections \ref{sec_detectors} and
\ref{sec_statistics}, respectively.

\texttt{threshold} specifies detection threshold(s) as either a scalar
(applied to all statistics) or a numeric vector (one threshold per
statistic). Setting \texttt{threshold\ =\ Inf} disables early stopping
and returns the full statistic traces (useful for thresholding) and
changepoint estimates at every time step. With finite threshold(s), the
algorithm terminates at the first time any statistic exceeds its
corresponding threshold, mirroring the online loop's stopping behaviour.
For multivariate detectors, the implementation computes multiple
statistics (one per projection or per configured statistic) and triggers
a detection when any of these crosses its threshold.

For example, the equivalent of the Gaussian change-in-mean online loop
described in the previous section can be achieved in the offline mode by
calling \texttt{focus\_offline} with the same cost function parameters:

\begin{verbatim}
R> result_offline <- focus_offline(Y, threshold = Inf,
+                                  type = "univariate",
+                                  family = "gaussian")
\end{verbatim}

\begin{figure}

\centering{

\includegraphics{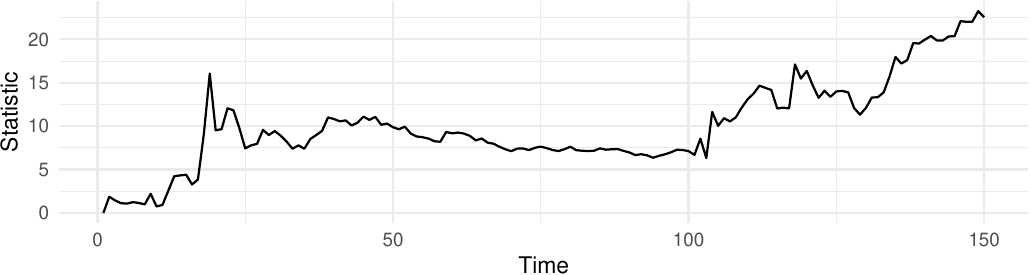}

}

\caption{\label{fig-r-offline}The trace of the statistics for one
simulated time series using the R \texttt{focus\_offline} interface.}

\end{figure}%

\section{Common Use Cases of Focus}\label{sec_use_cases}

Below are a few common examples showing how to create and update
different types of detectors. Again, for brevity, we focus on the
\proglang{R} interface, but the same examples can be easily adapted to
the \proglang{Python} interface as described above.

\subsection{Constrained Detection with One-sided
Detectors}\label{sec_constrained}

There may be some applications where one is only interested in detecting
changes in a specific direction, for example, only increases in the mean
of a signal. To create a one-sided detector, one can simply set the
\texttt{side} argument of \texttt{detector\_create} to either
\texttt{\textquotesingle{}right\textquotesingle{}} (to detect increases)
or \texttt{\textquotesingle{}left\textquotesingle{}} (to detect
decreases). We can show how to implement a detector with an example.

To choose an appropriate detection threshold, a common approach is to
use Monte Carlo simulation
\citep[as in][amongst many]{chen2022high, pishchagina2025online}. There
are various statistical criteria that can be used to specify a
threshold. Below we show how \texttt{focus\_offline} can be used to find
a threshold that gives a required false positive probability over a
fixed length of data. This example shows how using the
\texttt{threshold=Inf} setting enables us to do this efficiently without
re-running the algorithm for different thresholds.

Obtaining a threshold that gives a require false-positive rate for data
of length \(n\) is achieved by (i) generating many sequences of data
\(Y^{(k)}_{1:n}\) (\(k = 1, \ldots, K\)) under the null hypothesis of no
changepoints, (ii) computing the detection statistic
\(f_{\mathrm{GLR}}(Y^{(k)}_{1:n})\) for each sequence, and taking the
maximum over time \(\max_n f_{\mathrm{GLR}}(Y^{(k)}_{1:n})\) (as this is
the most extreme value that would trigger a detection), and finally
(iii) defining the threshold as the empirical quantile
\(\xi = F^{-1}_K(1 - \alpha)\), where \(F_K\) is the empirical
distribution of the maxima and \(\alpha\) is the desired false positive
rate. Specifically, to achieve a \(1\%\) false positive rate
(\(\alpha = 0.01\)), we take the 99th percentile of the null statistics
as our threshold; approximately 1\% of null sequences will then exceed
this threshold by chance.

To illustrate, here we obtain such Monte Carlo thresholds for a
one-sided detector on sequences of length \(10^5\):

\begin{verbatim}
library(furrr)
plan(multisession, workers = 4)
set.seed(42)
n_sim <- 500
runs <- future_map_dbl(
  1:n_sim, \(i) {
    Y_sim <- rnorm(1e5)
    res <- focus_offline(Y_sim, threshold = Inf,
      type = "univariate_one_sided", side = "right",
      family = "gaussian"
    )
    max(res$stat)
  },
  .options = furrr_options(seed = TRUE)
)
plan(sequential)
print(threshold_99 <- quantile(runs, 0.99))
\end{verbatim}

\begin{verbatim}
     99% 
31.23122 
\end{verbatim}

The code simulates 500 null sequences of independent Gaussian data,
computes the maximum detection statistic for each using
\texttt{focus\_offline} with \texttt{threshold\ =\ Inf} (returning all
statistics without early stopping), and then takes the 99th percentile
of these maxima as the threshold. We parallelise the simulations via the
\texttt{furrr} package to speed up computation. A similar
parallelisation strategy can be employed in the \proglang{Python}
interface using the \texttt{concurrent.futures} module.

We then generate some data with two changes, a decrease followed by an
increase, and we run the one-sided detector on this data with the
threshold obtained above:

\begin{verbatim}
R> set.seed(123)
R> Y <- c(rnorm(1e4, mean = 0), rnorm(1e4, mean = -1), rnorm(1e4, mean = 1))
R> det_one_sided <- detector_create(type = "univariate_one_sided",
+                                   side = "right")
R> for (i in seq_along(Y)) {
+    det_one_sided <- det_one_sided |> detector_update(Y[i])
+    result <- det_one_sided |> get_statistics(family = "gaussian")
+    if (result$stat > threshold_99)
+      break
+  }
R> sprintf("Changepoint detected at time %d", i)
\end{verbatim}

\begin{verbatim}
[1] "Changepoint detected at time 20016"
\end{verbatim}

We notice how we ignored the first change at time 10000, which
corresponds to a decrease in the mean, whilst we detected the second
change at time 20000.

\subsection{Multivariate Detectors}\label{sec_multi_detectors}

For multivariate detectors of low-dimensional data (fewer than 5
dimensions), the method prunes via computing the vertices of a convex
hull obtained from all dimensions:

\begin{verbatim}
R> set.seed(42)
R> n <- 1000
R> D <- 3
R> Y <- matrix(rnorm(n * D), nrow = n, ncol = D)
R> Y[501:1000, ] <- Y[501:1000, ] + 1
R> 
R> det_mv <- detector_create(type = "multivariate")
R> stat_trace <- vector("numeric", length = n)
R> for (i in seq_len(nrow(Y))) {
+    det_mv <- det_mv |> detector_update(Y[i, ])
+    result <- det_mv |> get_statistics(family = "gaussian")
+    stat_trace[i] <- result$stat
+    
+  }
\end{verbatim}

\begin{figure}

\centering{

\includegraphics{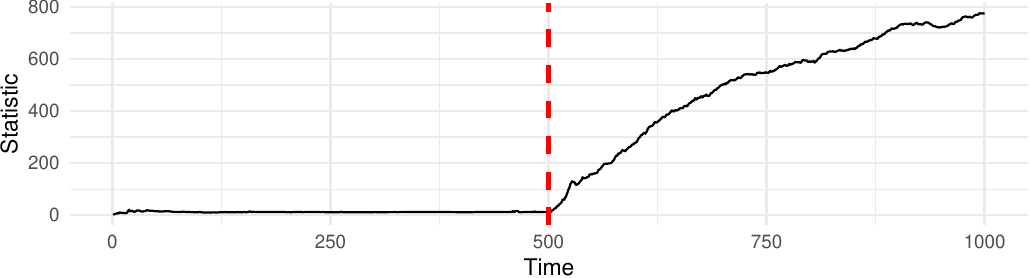}

}

\caption{\label{fig-multivariate}The trace of the statistics for one
simulated 3-dimensional multivariate time series using the R online
interface. Dashed red line illustrates the true change location.}

\end{figure}%

However, for higher-dimensional data, this can be computationally
expensive and may not be necessary for effective detection. As mentioned
in Section \ref{sec_update_pruning}, the package provides an
approximation of the hull by projecting the data onto overlapping
subsets of \(\tilde{d} \le 5\) coordinates, computing the convex hull on
each projection, and taking the union of the vertices. This can be
accessed via the \texttt{dim\_indexes} argument of
\texttt{detector\_create}, which specifies which set of coordinates to
project onto. To help generate these projection indexes, the function
\texttt{generate\_projection\_indexes} is provided:

\begin{verbatim}
dim_idx <- generate_projection_indexes(6, 2)
head(dim_idx, 2)
\end{verbatim}

\begin{verbatim}
[[1]]
[1] 0 1

[[2]]
[1] 1 2
\end{verbatim}

\begin{verbatim}
det_mv <- detector_create(type = "multivariate", dim_indexes = dim_idx)
\end{verbatim}

This particular example generates circular sequential projections of 2
dimensions from 6-dimensional data, producing 6 non-overlapping windows.
Each row of the output represents one projection and contains the
(0-indexed, since C++ uses 0-based indexing) coordinate indices to use.
In this instance, the first two projections use coordinates (0, 1) and
(1, 2), and the last wraps around to (5, 0).

As an illustration, we can measure how big the approximation is of using
the projections by calculating the statistics for a 6-dimensional
dataset with both the full hull and the projection-based approximation,
and compare both the run time and the resulting statistics:

\begin{verbatim}
R> system.time(
+    res_multi <- focus_offline(Y_multi, threshold = Inf,
+                               type = "multivariate", family = "gaussian")
+  )
\end{verbatim}

\begin{verbatim}
   user  system elapsed 
 10.284   0.123  10.409 
\end{verbatim}

\begin{verbatim}
R> system.time(
+    res_multi_approx <- focus_offline(Y_multi, threshold = Inf,
+                                      type = "multivariate", family = "gaussian",
+                                      dim_indexes = dim_idx)
+  )
\end{verbatim}

\begin{verbatim}
   user  system elapsed 
  0.166   0.000   0.166 
\end{verbatim}

\begin{verbatim}
R> all.equal(res_multi$stat, res_multi_approx$stat)
\end{verbatim}

\begin{verbatim}
[1] "Mean relative difference: 0.003782673"
\end{verbatim}

We observe that the projection-based approximation is significantly
faster than the full convex hull pruning, while still producing
extremely similar results.

\subsection{Anomaly Detection}\label{sec_anomaly_detection}

The i.i.d. detector allows for detecting anomalies (sometimes known as
epidemic changepoints) and more transient changes, while ignoring
longer, less intense variations. This is particularly useful when the
background rate is not constant but evolves slowly over time, as such
drift can otherwise lead to spurious detections when using a fixed
reference value. To address this, when \texttt{anomaly\_intensity} is
set to a positive value, change candidates are retained only if they
show sufficient ``signal intensity'' (i.e., the magnitude of the change
relative to the segment length is at least \texttt{anomaly\_intensity}).

To illustrate this, we generate data with a slowly varying background
\(\mu_t = 0.8 \sin(2 \pi t / 400)\), to which we add some Gaussian noise
and two transient anomalies (at times 100 and 800), as seen in Figure
\ref{fig-anomaly_detection}. We then run the detector with and without a
minimum change size. When computing the statistics, we set
\texttt{theta0\ =\ 0}, deliberately ignoring the slowly varying
background:

\begin{verbatim}
R> run_detector <- function(Y, anomaly_intensity = NULL) {
+    det <- detector_create(type = "univariate",
+                           anomaly_intensity = anomaly_intensity)
+    stat_trace <- numeric(length(Y))
+    for (i in seq_along(Y)) {
+      det <- det |> detector_update(Y[i])
+      res <- det |> get_statistics(family = "gaussian", theta0 = 0)
+      stat_trace[i] <- res$stat
+    }
+    stat_trace
+  }
R> stat_no_thresh <- run_detector(Y_anom, anomaly_intensity = NULL)
R> stat_thresh <- run_detector(Y_anom, anomaly_intensity = 1.5)
\end{verbatim}

In this example, the background mean varies smoothly with an amplitude
smaller than the imposed minimum change size. In Figure
\ref{fig-anomaly_detection}, below the data, we see the two
corresponding statistic traces \texttt{stat\_no\_thresh} and
\texttt{stat\_thresh}. We notice how, without the constraint, the
detector reacts to these slow variations. By setting
\texttt{anomaly\_intensity} larger than the maximum background
deviation, these are suppressed, while the transient anomalies are still
clearly identified.

\begin{figure}

\centering{

\includegraphics{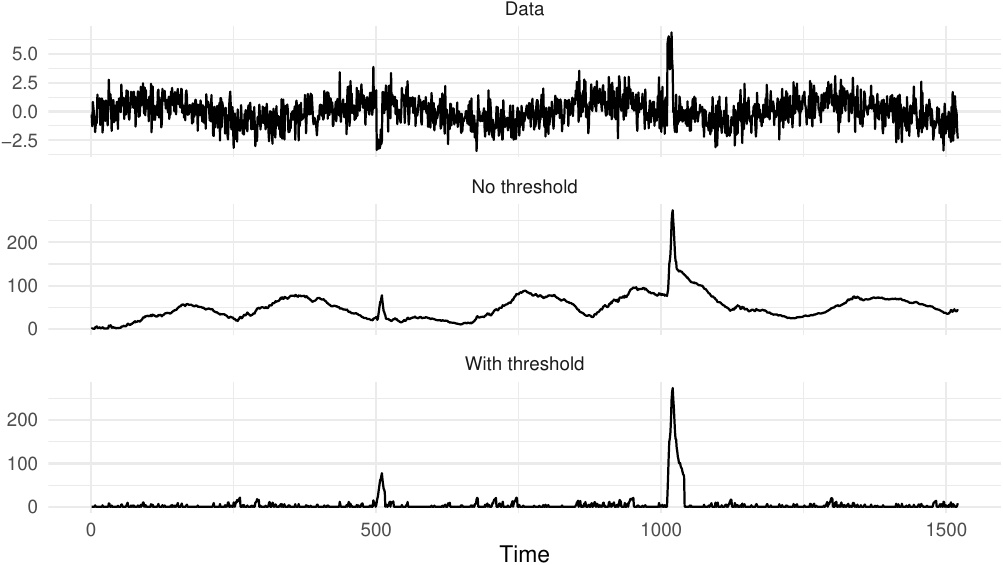}

}

\caption{\label{fig-anomaly_detection}Comparison of statistics with and
without minimum change size under a slowly varying background: Analysed
data (top); trace of test statistic without (middle) and with (bottom) a
minimum change size.}

\end{figure}%

\subsection{Non-Parametric Changepoint
Detection}\label{sec_nonparametric}

As a nonparametric detector, the package implements the \pkg{npfocus}
detector described in Section \ref{sec_nonparametric_method}. As
\pkg{npfocus} detects changes in the empirical cumulative density
function estimated across a set of quantiles, to create a non-parametric
detector, one should provide a vector of quantiles to the argument
\texttt{quantiles} of \texttt{detector\_create}. Quantiles normally can
be estimated from historical data, or from domain knowledge (for the
sake of the example below we will use quantiles from known
distributions). We can then create a non-parametric detector using these
quantiles. \pkg{npfocus} returns two statistics: one based on summing
over quantiles, and one based on taking the maximum over quantiles.

As an illustration, to demonstrate a comparison of the behaviour of the
two statistics, we consider two scenarios: (1) a change in location
using heavy-tailed \(t_2\) data (stored below in \texttt{Y1}), and (2) a
Gaussian with a change only in the tail (stored in \texttt{Y2}), where
values with absolute values greater than 2 are shifted up by 20. We then
run the non-parametric detector online on both data sets, and store the
trace of both statistics for comparison. In Figure \ref{fig-npfocus} we
plot the resulting statistic traces in both scenarios together with the
data.

\begin{verbatim}
R> q1 <- qt(seq(0.01, 0.99, length.out = 5), df = 2)
R> det_np1 <- detector_create(type = "npfocus", quantiles = q1)
R> stat_trace1 <- matrix(nrow = length(Y1), ncol = 2)
R> for (i in seq_along(Y1)) {
+    det_np1 <- det_np1 |> detector_update(Y1[i])
+    res <- det_np1 |> get_statistics(family = "npfocus")
+    stat_trace1[i, ] <- res$stat
+  }
R> q2 <- qnorm(seq(0.01, .99, length.out = 4))
R> det_np2 <- detector_create(type = "npfocus", quantiles = q2)
R> stat_trace2 <- matrix(nrow = length(Y2), ncol = 2)
R> for (i in seq_along(Y2)) {
+    det_np2 <- det_np2 |> detector_update(Y2[i])
+    res <- det_np2 |> get_statistics(family = "npfocus")
+    stat_trace2[i, ] <- res$stat
+  }
\end{verbatim}

\begin{figure}

\centering{

\includegraphics{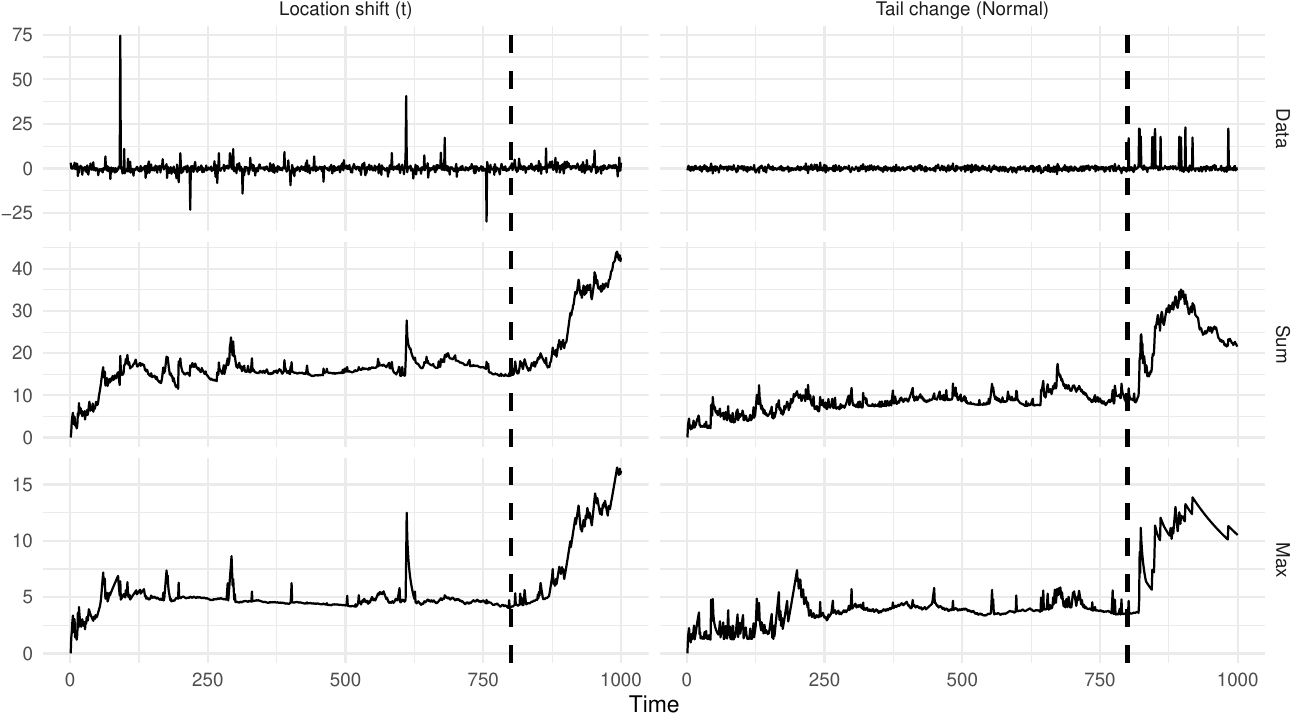}

}

\caption{\label{fig-npfocus}Comparison of sum and max statistics for two
different change scenarios. Left: change in location for \(t_2\) data.
Right: change only in the tail for normal data. Dashed line indicates
the true change location.}

\end{figure}%

\subsection{Autoregressive Changepoint Detection}\label{sec_ar}

To create an AR(p) detector, one must specify the AR coefficients
\(\rho = (\rho_1, \ldots, \rho_p)\) of the underlying process. These can
be estimated from historical data. Additionally, if the pre-change mean
is known, it can be provided to the detector for more efficient pruning,
though this must be specified at detector creation time via the
\texttt{mu0\_arp} argument.

As an example, we first generate data from an AR(2) process with a
change in mean. We start by specifying the AR coefficients and
generating some data:

\begin{verbatim}
R> set.seed(123)
R> ar_coefs <- c(0.7, -0.3)
R> Y_pre <- arima.sim(n = 5000, model = list(ar = ar_coefs), sd = 1)
R> Y_post <- 2 + arima.sim(n = 100, model = list(ar = ar_coefs), sd = 1)
R> Y <- c(Y_pre, Y_post)
\end{verbatim}

We then run the AR(\(p\)) detector online on this data. For this
illustration we estimate the AR coefficients from a probation period of
1000 observations, as in real-world applications these can be estimated
from historical data. We plot the resulting statistic trace together
with the data in Figure \ref{fig-arp_detection}.

\begin{verbatim}
R> ar_coefs <- ar.yw(Y_pre[1:1000], order.max = 5)$ar
R> det_arp <- detector_create(type = "arp", rho = ar_coefs, mu0_arp = 0)
R> stat_trace <- numeric(length(Y))
R> 
R> for (i in 1001:length(Y)) {
+    det_arp <- det_arp |> detector_update(Y[i])
+    result <- det_arp |> get_statistics(family = "arp")
+    stat_trace[i] <- result$stat
+  }
R> tail(result$stat)
\end{verbatim}

\begin{verbatim}
[1] 183.7202
\end{verbatim}

\begin{figure}

\centering{

\includegraphics{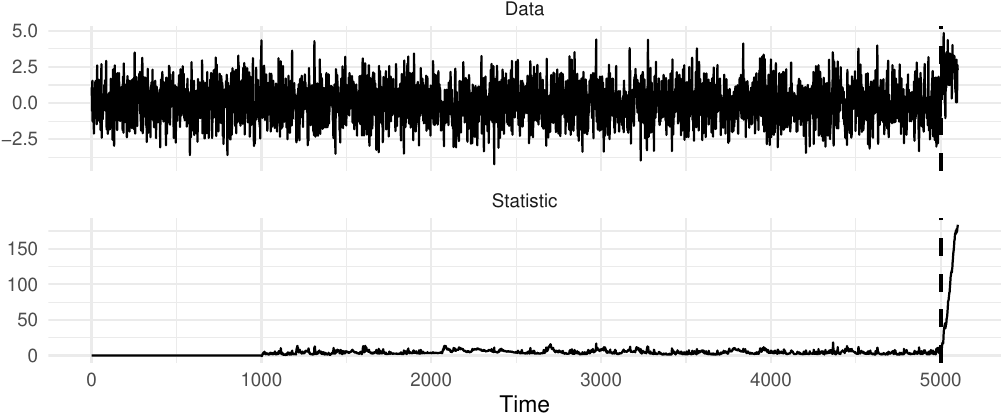}

}

\caption{\label{fig-arp_detection}The AR(2) process with mean shift and
the resulting detection statistic using the R online interface. Dashed
line illustrates the true change location.}

\end{figure}%

\section{Some Examples of Real-Time
Applications}\label{sec_real_time_applications}

In this section, we illustrate few real-world applications of the
package. The NBA application (Section \ref{sec_meanvar}) is implemented
in \proglang{R}, while the gamma-ray burst detection (Section
\ref{sec_grb}) and spike inference (Section \ref{sec_calcium}) examples
are implemented in \proglang{Python}.

\subsection{Implementing a Custom Cost Function: Change in Mean and
Variance on NBA Plusminus Scores}\label{sec_meanvar}

In this section, we illustrate how to implement a custom cost function
for a bivariate model to detect a change in the Cleveland Cavaliers (an
NBA team) Plus-Minus score from the season 2010-11 to the 2017-18 season
as in \citep{pishchagina2025online}.

As illustrated in \citep{pishchagina2025online}, simulating data to
calibrate a simple Gaussian change in mean model for this NBA dataset is
tricky. This is because a small change in the variance of the reference
distribution changes the threshold and, as a consequence, can increase
the detection time or leads to false detections. A robust alternative to
this problem is to consider a Gaussian change in mean and variance
model. Essentially, we let the data decide what the variance is before
and after the change. Because we look at the score difference, we might
very well have successive matches with exactly the same score
difference. In that case, the maximum likelihood estimator of the
variance is 0 and the likelihood would be infinite. To overcome this, we
consider, as in \citep{pishchagina2025online}, a lower bound on the
variance of 1.

As explained just after equation (6) in \citep{pishchagina2025online},
the hull argument holds even when the set of parameters is constrained.
In terms of code, it means that we only need to implement an alternative
\texttt{get\_statistics} function to run this test, and the detector
update will be the one of the standard change in mean and variance
model.

In detail, for a Gaussian change in mean and variance we need to get the
points on the hull of all \((t, \sum_{i=1}^t x_i, \sum_{i=1}^t x^2_i)\)
for \(t<n.\) Now, assuming \(t\) is on the hull, we first estimate the
variance on the segment before the change (going from \(1\) to \(t\)).
Our variance estimator is simply

\[
\hat{\sigma}^2_{1:t} = \min \left\{ \frac{1}{t}\left[\sum_{i=1}^t x^2_i - \frac{1}{t} \left(\sum_{i=1}^t x_i\right)^2\right], 1 \right\}
\]

and from that we get the likelihood as

\[
t\log(\hat{\sigma}^2_{1:t}) + \frac{\sum_{i=1}^t x^2_i - (1/t) (\sum_{i=1}^t x_i)^2}{\hat{\sigma}^2_{1:t}}.
\] In code we get, for each segment defined by the candidate \(\tau\)
and its corresponding summary statistics, the following function to
compute the likelihood:

\begin{verbatim}
R> get_seg_meanvar <- function(info_seg, min_var=1){
+    if(info_seg[1] > 0){
+      est_mse <-  info_seg[3] - info_seg[2]^2/info_seg[1]
+      est_var <- est_mse/info_seg[1]
+      if(est_var < min_var) est_var <- min_var
+      return(info_seg[1]*log(est_var) + est_mse/est_var)
+      } else { return(Inf) }
+  }
\end{verbatim}

We proceed similarly to obtain the likelihood of the right segment,
going from \(t+1\) to \(n\) by replacing in the formula \(t\) by
\(n-t\), \(\sum_{i=1}^t x_i\) by \(\sum_{i=1}^n x_i - \sum_{i=1}^t x_i\)
and \(\sum_{i=1}^t x^2_i\) by
\(\sum_{i=1}^n x^2_i - \sum_{i=1}^t x^2_i.\) This leads to the full cost
function:

\begin{verbatim}
R> library(purrr)
R> get_stat_meanvar <- function(det_ptr, min_var=1) {
+    mat <- detector_candidates(det_ptr)
+    n <- detector_info_n(det_ptr)
+    all <- detector_info_sn(det_ptr)
+    stats <- map_dbl(seq_along(mat$tau)[-1], \(i) {
+      info_seg_left <- c(mat$tau[i], mat$st[[i]])
+      
+      info_seg_right <- c(n, all) - info_seg_left
+      return(
+          - get_seg_meanvar(info_seg_left, min_var) -
+          get_seg_meanvar(info_seg_right, min_var)  +
+          get_seg_meanvar(info_seg_left + info_seg_right, min_var)
+      )
+    })
+    max(stats)
+  }
\end{verbatim}

Then for training and testing we first recover all the Cleveland
Cavaliers game from 2000 to 2025.

Then to calibrate our threshold, we (1) resample uniformly at random 640
games prior to 2010 (each NBA regular season is made up of 64 games),
(2) compute the test statistics at all time-points, and (3) recover the
maximum.

We repeated this a thousand times and picked our threshold as the 99\%
quantile.

\begin{verbatim}
R> library(furrr)
R> plan(multisession, workers = 8)
R> 
R> dat_past <- dat %>% filter(yearSeason <= 2010, typeSeason == "Regular Season")
R> 
R> y <- future_map(1:10^3, \(i) {
+    plusminus_team <- sample(dat_past$plusminusTeam, 1000, replace = TRUE)
+    dat_for_focus <- rbind(plusminus_team, plusminus_team^2)
+    
+    det <- detector_create(type = "multivariate")
+    stat_record <- vector("numeric", length = ncol(dat_for_focus))
+    for (i in 1:ncol(dat_for_focus)) {
+      detector_update(det, dat_for_focus[, i])
+      stat_record[i] <- get_stat_meanvar(det)
+    }
+    # max_stat_record
+    max(stat_record)
+  }, .options = furrr_options(seed = TRUE))
R> results <- do.call(c, y)
R> 
R> threshold_99 <- quantile(results, 0.99)
R> 
R> print(threshold_99)
\end{verbatim}

\begin{verbatim}
     99% 
32.79621 
\end{verbatim}

We then test the detector on the data from 2011 onwards:

\begin{verbatim}
R> dat_recent <- dat %>% filter(yearSeason > 2010, typeSeason == "Regular Season")
R> Y_test <- rbind(dat_recent$plusminusTeam, dat_recent$plusminusTeam^2)
R> det <- detector_create(type = "multivariate")
R> stat_record <- vector("numeric", length = ncol(Y_test))
R> for (i in 1:ncol(Y_test)) {
+    detector_update(det, Y_test[, i])
+    stat_record[i] <- get_stat_meanvar(det, min_var=1)
+  }
\end{verbatim}

And plot both data and statistics (as above) in Figure
\ref{fig-nba-application}.

\begin{figure}

\centering{

\includegraphics{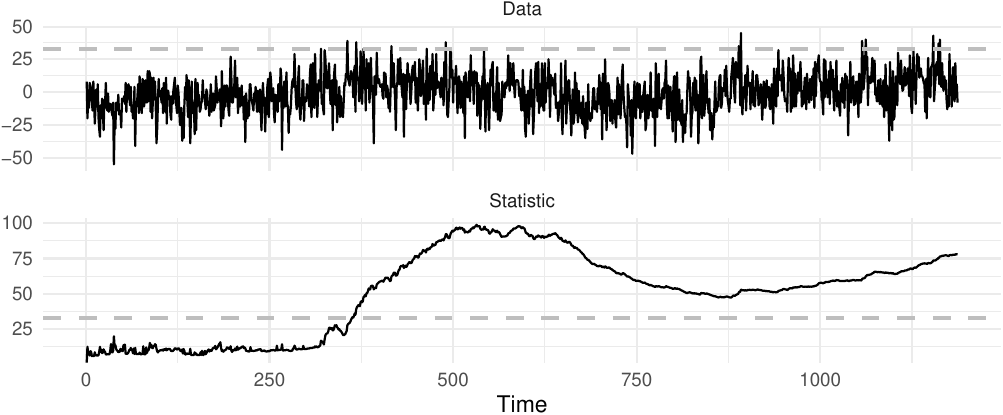}

}

\caption{\label{fig-nba-application}The NBA plus-minustest data (top)
and the resulting test statistics (bottom). Dashed gray line indicates
the Monte Carlo Threshold.}

\end{figure}%

As discussed in \citep{pishchagina2025online} and
\citep{shin2023detectors} we detect a change in the 2014-2015 season
corresponding to the return of James Lebron in the Cleveland Cavaliers.

\subsection{Gamma-Ray Burst Detection}\label{sec_grb}

Gamma-ray bursts (GRBs) are intense flashes of gamma rays. Their
detection is crucial for understanding astrophysical phenomena, and fast
and accurate online detection is essential for follow-up observations.
Here, we demonstrate how the package can accurately identify the
significance of a gamma-ray burst in publicly available data from the
Fermi-GBM (Gamma-ray Burst Monitor) instrument \citep{meegan2009fermi}.
As demonstrated in \citep{ward2022poisson, dilillo2024gamma} the
advantage of using a \pkg{focus}-based trigger algorithm compared to
traditional thresholding methods, is that it identifies the most
significant interval in the gamma-ray burst as it happens, without
requiring knowledge of the true duration in advance and any
computational tradeoffs.

The data we use is from the Fermi-GBM triggers catalogue for burst
GRB250814432. The exact data file can be found on
\href{https://heasarc.gsfc.nasa.gov/FTP/fermi/data/gbm/triggers/2025/bn250814432//current/}{this
page}. The file contains time-to-event data spanning 10 minutes. We bin
it into 10 millisecond intervals and examine the period around the
gamma-ray burst. To detect the burst online, we initialise a detector
and stream the data through it, updating the detector with each new
observation. For the sake of this example, we estimate the pre-change
parameter \(\theta_0\) over a fixed length of a minute prior to the
start of the sequence.

\begin{verbatim}
>>> data_window = data_full['2025-08-14 10:01':'2025-08-14 10:02']
... theta0_est = np.mean(data_full['2025-08-14 10:00':'2025-08-14 10:00'])
... detector = focus_cpt.Detector(type="univariate_one_sided",
...                               side="right",
...                               anomaly_intensity = 1.5)
... stats_list = []
... for i, value in enumerate(data_window):
...     detector.update(value)
...     result = detector.get_statistics(family="poisson", theta0=theta0_est)
...     stats_list.append(result['stat'])
... stat_array = np.array(stats_list)
... significance = np.sqrt(2 * stat_array)
\end{verbatim}

The significance transformation, as mentioned in Equation 1 of
\citep{dilillo2024gamma}, is necessary to obtain a more interpretable
measure of the detection statistic for physicists that can be directly
compared to standard significance thresholds (e.g., 5-sigma). The plot
in Figure \ref{fig-gamma-ray-burst} shows the count rate and the running
test statistic in units of \(\sigma\), with a vertical line indicating
the first point at which the statistic exceeds the 5-sigma threshold.

\begin{figure}

\centering{

\includegraphics{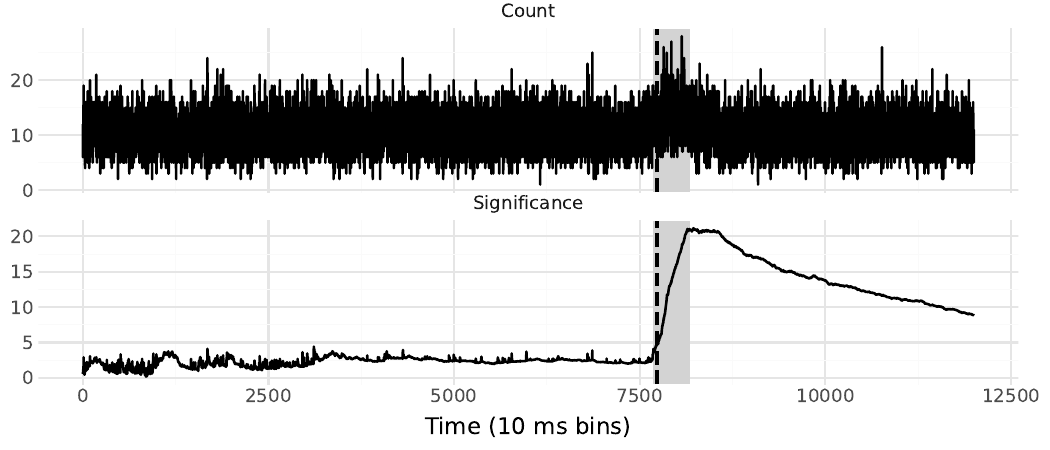}

}

\caption{\label{fig-gamma-ray-burst}Online detection of a gamma-ray
burst from Fermi-GBM data. Top: radiation count rate; Bottom: running
test statistic in units of \(\sigma\) with a vertical line indicating
the first point exceeding the 5-sigma threshold. The shaded region
highlights the interval of most significant detection.}

\end{figure}%

\subsection{Constrained Up-Down Model for Spike
Inference}\label{sec_calcium}

We now present an application to calcium imaging data, based on the
\texttt{spikefinder} challenge data
(\url{https://github.com/codeneuro/spikefinder}). This data measures
fluorescence, relating to the intensity of calcium ions in neurons, and
is sampled at a high-frequency of 100 Hz. In calcium imaging, each
neuron firing (spike) produces a sharp upward jump followed by an
exponential decay. Changepoint detection approaches have been
successfully employed to analyse calcium imaging data in the offline
setting, yielding computationally efficient algorithms for spike
inference \citep{jewell2018exact, jewell2020fast, fleming2021inferring}.
However, there is increasing interest in performing equivalent analyses
in real time. In-vivo calcium imaging now enables closed-loop
experiments in which neural activity is decoded and available online
\citep{hira2024closed, zhang2025real}. These applications require
low-latency processing of fluorescence signals, whereas most existing
pipelines operate offline and incur substantial delays
\citep{chen2023fpga}, motivating the need to perform spike inference for
real-time streaming data. The original data is a trace of fluorescence
intensity over time, which we will use to infer the underlying spike
times; see Figure \ref{fig-calcium_trace}. Although the
\texttt{spikefinder} dataset is offline in nature, the use of in-vitro
recordings provides a controlled validation setting, as the true spike
times are known from simultaneous electrophysiological recordings. This
allows us to tune the thresholds of the two detectors and evaluate their
performance.

To analyse this sequence in a sequential fashion via our package, we use
two one-sided detectors simultaneously: one watching for up-changes
(spikes) and one for down-changes (decay onset), resetting both whenever
either detects a change. The rationale behind this is that only
up-changes are of interest, and that given the decay process, it would
be beneficial to have two separate thresholds to better control the
process. Below, we define a helper function that implements the
two-detectors logic, as well as records the detected spike times, the
change times, and the statistics of both detectors at each time step:

\begin{verbatim}
>>> def run_detector(Y, time, r_threshold, l_threshold):
...     r_det = Detector(type='univariate_one_sided', side='right')
...     l_det = Detector(type='univariate_one_sided', side='left')
...     up_detections = []
...     stp_times     = []
...     stp_types     = []
...     for t, y in enumerate(Y):
...         r_det.update(y)
...         l_det.update(y)
...         r_s = r_det.get_statistics(family='gaussian')['stat']
...         l_s = l_det.get_statistics(family='gaussian')['stat']
...         cpt_found = False
...         if r_s > r_threshold:
...             cpt_found, cpt_type = True, 'right'
...         elif l_s > l_threshold:
...             cpt_found, cpt_type = True, 'left'
...         if cpt_found:
...             if cpt_type == 'right':
...                 up_detections.append(t)
...             stp_times.append(time[t])
...             stp_types.append(cpt_type)
...             r_det = Detector(type='univariate_one_sided', side='right')
...             l_det = Detector(type='univariate_one_sided', side='left')
...     return {
...         'up_detections': np.array(up_detections),
...         'stp_times'    : stp_times,
...         'stp_types'    : stp_types,
...     }
\end{verbatim}

To give a brief intuition on the function above, at each time step, we
update both detectors with the new observation, and we check if either
of them has detected a change by comparing their statistics to the
corresponding thresholds: \texttt{r\_threshold}, for up-changes, and
\texttt{l\_threshold},for down-changes. If a change is detected, we
record the time and type of the change, and we reset both detectors. The
up-changes are recorded separately as they correspond to the spikes we
want to detect. The two thresholds are passed as arguments to the
function, as this will allow us to tune them via cross validation. To do
so, we create a 30/70 \% train test split. Then, both thresholds are
selected jointly by 5-fold cross-validation on the training set. As a
criterion for the goodness of fit, we choose to minimise van Rossum
distance between the detected and true spike trains. This is employed as
a detection during the calcium decay, e.g.~when we are far from any true
spike, is penalised more than a miss. Furthermore near-burst detections
are scored continuously rather than with a hard cutoff, so the metric
reflects the biological reality that not all spikes in a burst may be
detected, but detecting some of them is still useful
\citep{jewell2020fast}.

Using the CV-selected thresholds \texttt{best\_r\_threshold} and
\texttt{best\_l\_threshold}, we evaluate on the held-out test set by
calling \texttt{run\_detector} directly. To illustrate the results, we
plot the results in Figure \ref{fig-calcium_trace}, showing the detected
changepoints and the trace of the right-detector statistic.

\begin{verbatim}
>>> res_test = run_detector(Y_test, time_test,
...                         best_r_threshold, best_l_threshold)
\end{verbatim}

\begin{figure}

\centering{

\includegraphics{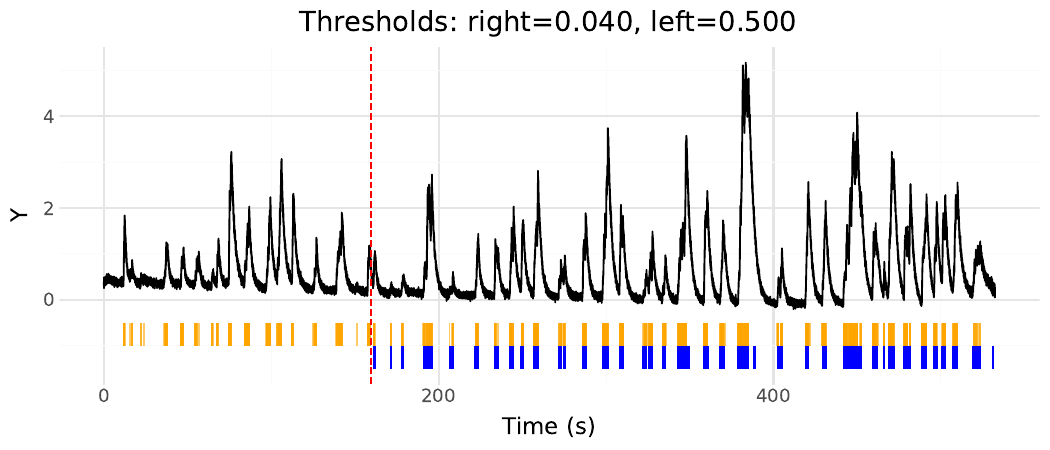}

}

\caption{\label{fig-calcium_trace}The trace of fluorescence intensity
over time from a calcium imaging recording of our example. The orange
segments indicate the true spike times obtained from simultaneous
electrophysiological recordings, while the blue segments the detected up
changes. Dashed red line indicates the start of the test period.}

\end{figure}%

\section{Acknowledgments}\label{acknowledgments}

EPSRC's support via grant UKRI2698 (Fearnhead, Eckley and Romano) and
EP/Z531327/1 (Eckley and Fearnhead) is acknowledged. Fan also
acknowledges the financial support of EPSRC and BT via an Industrial
CASE award.

\section*{References}\label{references}
\addcontentsline{toc}{section}{References}

\renewcommand{\bibsection}{}
\bibliography{bibliography.bib}

\newpage{}

\newpage
\appendix
\begin{center}
{\large\bf APPENDIX}
\end{center}

\section{The Python Interface: Differences and
Similarities}\label{sec_python_interface}

The \proglang{Python} package \texttt{{[}focus-cpt{]}\{.pkg\}} provides
a clean, Pythonic interface to the same underlying C++ implementation as
the \proglang{R} package. The core philosophy is to mirror the
\proglang{R} interface as closely as possible, while leveraging
\proglang{Python} idioms and conventions where appropriate.

The primary structural difference between the two interfaces lies in the
handling of the detector object. In \proglang{R}, the detector is an
external pointer that must be passed to functions like
\texttt{detector\_update} and \texttt{get\_statistics}. In
\proglang{Python}, the detector is an instance of the \texttt{Detector}
class, and operations are performed as methods on this object. This is
more idiomatic in \proglang{Python} and aligns with the object-oriented
philosophy of the language.

Specifically:

\begin{itemize}
\item
  \textbf{Initialization}: In \proglang{R},
  \texttt{detector\_create(type\ =\ "univariate")} creates a detector.
  In \proglang{Python}, this is \texttt{Detector(type="univariate")}.
\item
  \textbf{Updating}: In \proglang{R}, \texttt{detector\_update(det,\ y)}
  updates the detector. In \proglang{Python}, this is
  \texttt{det.update(y)}.
\item
  \textbf{Computing statistics}: In \proglang{R},
  \texttt{get\_statistics(det,\ family\ =\ "gaussian")} computes
  statistics. In \proglang{Python}, this is
  \texttt{det.get\_statistics(family="gaussian")}.
\item
  \textbf{Extracting information}: Functions like
  \texttt{detector\_info\_n}, \texttt{detector\_info\_sn}, and
  \texttt{detector\_cands\_len} in \proglang{R} are available as methods
  \texttt{get\_n()}, \texttt{get\_sn()}, and
  \texttt{get\_n\_candidates()} in \proglang{Python}. Similarly,
  \texttt{detector\_candidates} is available as the method
  \texttt{get\_candidates()}.
\end{itemize}

Beyond the object-oriented interface, the \proglang{Python} bindings
follow standard \proglang{Python} conventions:

\begin{itemize}
\item
  \proglang{Python} uses \texttt{None} instead of R's \texttt{NULL}. For
  example, \texttt{get\_statistics(family="gaussian",\ theta0=None)} is
  equivalent to the R call
  \texttt{get\_statistics(det,\ family\ =\ "gaussian",\ theta0\ =\ NULL)}.
\item
  Numeric vectors and matrices in \proglang{Python} are typically
  represented as NumPy arrays. The package should automatically handle
  conversion of lists and other array-like objects to NumPy arrays. Note
  that for univariate detectors, the \texttt{focus\_offline} function
  returns a 2D statistic array with shape \texttt{(n\_obs,\ 1)} rather
  than a 1D array, which may need to be flattened before further
  processing (see the offline example below).
\item
  Functions return dictionaries (\proglang{Python}'s equivalent to
  \proglang{R}'s lists) rather than named lists. For example,
  \texttt{det.get\_statistics(family="gaussian")} returns a dictionary
  with keys \texttt{\textquotesingle{}stopping\_time\textquotesingle{}},
  \texttt{\textquotesingle{}changepoint\textquotesingle{}}, and
  \texttt{\textquotesingle{}stat\textquotesingle{}}.
\end{itemize}

\subsubsection{Example Usage in Python}\label{example-usage-in-python}

To better highlight the \proglang{Python} syntax, we replicate the
example at the beginning of the section, noting that results will not be
exactly the same due to differences in random number generation from
numpy and \proglang{R}:

\begin{verbatim}
>>> import numpy as np
... from focus_cpt import Detector
... 
... 
... np.random.seed(42)
... Y = np.concatenate([np.random.randn(100), np.random.randn(50) + 1])
... 
... det = Detector(type="univariate")
... 
... for i, y in enumerate(Y):
...     det.update(y)
...     result = det.get_statistics(family="gaussian")
...     
...     if result['stat'] > 20:
...         print(
...             f"Changepoint detected at time {i}: "
...             f"estimated changepoint at {result['changepoint']}"
...         )
...         break
\end{verbatim}

\begin{verbatim}
Changepoint detected at time 116: estimated changepoint at 104
\end{verbatim}

The offline interface is also available in \proglang{Python} with the
same signature as in \proglang{R}, again with the \texttt{None}
convention for optional parameters. For example:

\begin{verbatim}
>>> import pandas as pd
... from plotnine import *
... from focus_cpt import focus_offline
... 
... result_offline = focus_offline(Y,
...                                 threshold=np.inf,
...                                 type="univariate",
...                                 family="gaussian")
... 
... stat = result_offline['stat'].flatten()
... 
... df = pd.DataFrame({"time": range(1, len(Y) + 1), "stat": stat})
... (
...   ggplot(df) +
...     aes(x="time", y="stat") +
...     geom_line() +
...     labs(x="Time", y="Statistic") +
...     theme_minimal() +
...     theme(figure_size=(7, 3))
... )
\end{verbatim}

\begin{figure}[H]

\centering{

\includegraphics{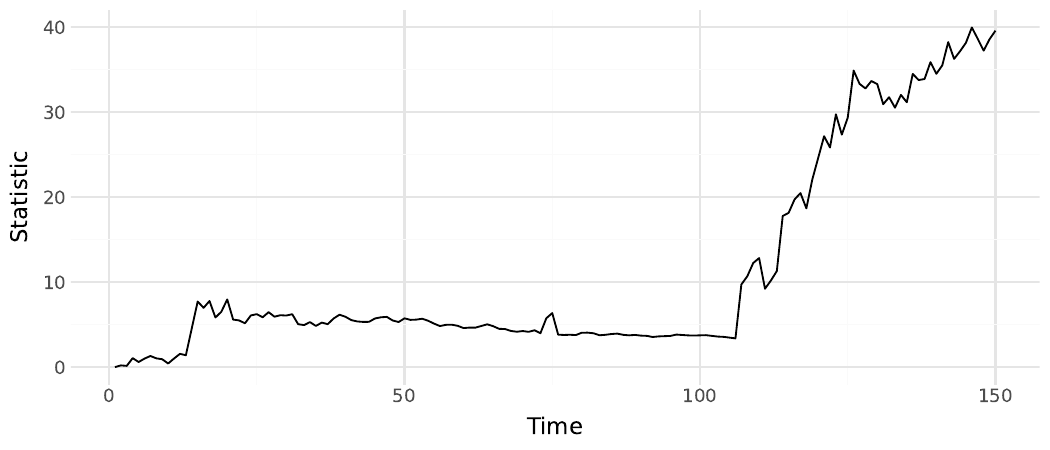}

}

\caption{\label{fig-python-offline}The trace of the statistics for one
simulated time series using the Python interface.}

\end{figure}%

\end{document}